\documentclass[a4paper, 11pt]{article}

\usepackage[T1]{fontenc}
\usepackage[utf8]{inputenc}
\usepackage[english]{babel}

\usepackage{amsthm}
\usepackage{amsmath}
\usepackage{amssymb}

\usepackage{authblk}
\usepackage[margin=2.5cm]{geometry}
\usepackage[
	pdftitle={An Improved Upper Bound for the Ring Loading Problem},
	pdfsubject={Ring Loading Problem},
	pdfauthor={Karl Däubel},
	pdfkeywords={Ring Loading Problem, SONET ring, load balancing, unsplittable flow}
]{hyperref}
\usepackage{enumerate}
\usepackage{booktabs}
\usepackage{multirow}
\usepackage{nameref}
\usepackage{csquotes}
\usepackage{graphicx}
\usepackage{subfig}

\usepackage[capitalise]{cleveref}

\usepackage{todonotes}

\usepackage{tikz}
\usetikzlibrary{decorations.pathreplacing}

\def\radius{2.5cm}
\def\radiusoffset{5pt}
\def\delangle{12}

\definecolor{node_col}{RGB}{0,0,180}
\definecolor{demand_col}{RGB}{180,0,0}

\definecolor{split_col_u}{RGB}{255,0,0}

\definecolor{split_col_v}{RGB}{0,255,0}

\definecolor{load_col}{RGB}{0,0,255}

\tikzset{
mynode/.style={draw, circle, fill = node_col, text height = 0pt, text width = 0pt, text depth = 0pt, inner sep = 0pt, minimum size = 5pt},
split_arc/.style={draw, very thick,-{latex}},
mydemand/.style={thick, draw = demand_col},
}

\DeclareMathOperator*{\argmin}{arg\,min}

\newtheorem{theorem}{Theorem}

\newtheorem{definition}{Definition}
\newtheorem{lemma}{Lemma}
\newtheorem{observation}{Observation}

\newcommand{\NN}{\mathbb{N}}

\newcommand{\RR}{\mathbb{R}}

\author{Karl D\"aubel}
\title{An Improved Upper Bound for the Ring Loading Problem}
\affil{Institut f\"ur Mathematik, Technische Universit\"at Berlin, Germany \\
  {\small \texttt{daeubel@math.tu-berlin.de}}}
\date{}

\begin{document}
\maketitle

\begin{abstract}
The \emph{Ring Loading Problem} emerged in the 1990s to model an important special case of telecommunication networks (SONET rings) which gained attention from practitioners and theorists alike. Given an undirected cycle on $n$ nodes together with non-negative demands between any pair of nodes, the \emph{Ring Loading Problem} asks for an unsplittable routing of the demands such that the maximum cumulated demand on any edge is minimized. Let $L$ be the value of such a solution. In the relaxed version of the problem, each demand can be split into two parts where the first part is routed clockwise while the second part is routed counter-clockwise. Denote with $L^*$ the maximum load of a minimum split routing solution. In a landmark paper, Schrijver, Seymour and Winkler~\cite{MR1612841} showed that $L \leq L^* + 1.5D$, where $D$ is the maximum demand value. They also found (implicitly) an instance of the \emph{Ring Loading Problem} with $L = L^* + 1.01D$. Recently, Skutella~\cite{MR3463048} improved these bounds by showing that $L \leq L^* + \frac{19}{14}D$, and there exists an instance with $L = L^* + 1.1D$. We contribute to this line of research by showing that $L \leq L^* + 1.3D$. We also take a first step towards lower and upper bounds for small instances.
\end{abstract}

\section{Introduction}

Given an undirected cycle on $n$ nodes together with non-negative demands between any pair of nodes, the \emph{Ring Loading Problem} asks for an unsplittable routing of the demands such that the maximum cumulated demand on any edge is minimal. Formally, we are given a graph $G = \left(V, E\right)$ with nodes $V = \left[n\right] := \left\{1, \ldots, n\right\}$, edges $\left\{i, i+1\right\}$ for each $i \in V$, where we assume throughout the paper that $\left\{n,n+1\right\} := \left\{n,1\right\}$, and demands for each pair of nodes $i < j$ of value $d_{i,j} \geq 0$. By a slight abuse of notation, we refer to both the demand from $i$ to $j$ and its value as $d_{i,j}$. An unsplittable solution decides for each demand whether it should be routed clockwise, sending all of its value along the path $\left\{i,i+1,\ldots,j\right\}$, or counter-clockwise, sending all of its value along the path $\left\{i,i-1,\ldots,1,n,\ldots,j\right\}$. The \emph{load} of an edge, for a given solution, is the sum of all demand values that are routed on paths that use the edge. We call the maximum load on any edge of the ring the \emph{load} of the solution. The problem is to find an unsplittable routing that minimizes the load. We denote with $L$ the load of such an optimal unsplittable solution. See~\cref{fig:intro} for an example.

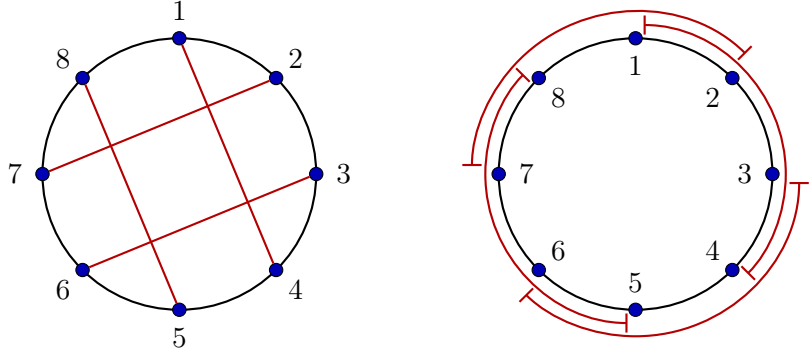
\begin{figure}
\centering
\begin{tikzpicture}[label_style/.style={circle, inner sep = 0pt, label distance = 3pt}]
\def\n{8}
\def\r{1.8}

\draw[thick] (0,0) circle(\r cm);

\foreach \i in {1,2,...,\n} {
  \node[mynode, label={[label_style]{90 + (\i -1) * -360 / \n}:\i}] (\i) at ({90 + (\i -1) * -360 / \n}:\r cm) {};
}

\foreach \i/\j in {1/4,2/7,3/6,5/8} {
  \draw[demand_col, thick] (\i) to[] (\j);
}

\begin{scope}[xshift=6cm]
\draw[thick] (0,0) circle(\r cm);

\foreach \i in {1,2,...,\n} {
  \node[mynode, label={[label_style]{-90 + (\i -1) * -360 / \n}:\i}] (\i) at ({90 + (\i -1) * -360 / \n}:\r cm) {};
}

\foreach \i/\j/\offset in {1/4/5,5/8/5,3/6/10} {
  \draw[demand_col, thick, shorten >= 3pt, shorten <= 3pt, |-|] ({90 + (\i -1) * -360 / \n}:\r cm + \offset pt) arc ({90 + (\i -1) * -360 / \n}:{90 + (\j -1) * -360 / \n}:\r cm + \offset pt);
}
\foreach \i/\j/\offset in {2/7/10} {
  \draw[demand_col, thick, shorten >= 3pt, shorten <= 3pt, |-|] ({90 + (\i -1) * -360 / \n}:\r cm + \offset pt) arc ({90 + (\i -1) * -360 / \n}:{360 + 180 - 90 + (\j -1) * -360 / \n}:\r cm + \offset pt);
}
\end{scope}
\end{tikzpicture}
\caption{An instance of the \emph{Ring Loading Problem} on $8$ nodes and $4$ non-zero demands with $d_{1,4} = d_{2,7} = d_{3,6} = d_{5,8} = 1$ (left) together with an optimum unsplittable routing of load $2$ (right).}
\label{fig:intro}
\end{figure}

The problem was introduced by Cosares and Saniee~\cite{Cosares1994} to mathematically model survivable networks with respect to the emerging standard of synchronous optical networks (SONET). The underlying structure to this technology, the SONET ring, is a set of network nodes and links that are arranged in a cycle. In this way, even in the event of a link failure, most of the traffic could be recovered. See~\cite{20262,Goralski:2002:SON:515627,book_sonet} for further resources on technical details. To the best of our knowledge, Cosares and Saniee~\cite{Cosares1994} also established the name \emph{Ring Loading Problem}. They further showed via a reduction of the \emph{Partition Problem} that the problem is NP-hard and provided an algorithm that returns an unsplittable solution with load at most $2L$. Using a result from Schrijver et al.~\cite{MR1612841}, Khanna~\cite{6768080} showed that there exists a PTAS, i.e. a class of poly-time algorithms that return a solution with load at most $\left(1 + \varepsilon\right)L$, for each fixed $\varepsilon > 0$. If all non-zero demands have the same value, Frank~\cite{FRANK1985164} showed that the \emph{Ring Loading Problem} can be solved in polynomial time.

Although a PTAS for the \emph{Ring Loading Problem} exists, there remain unsolved problems that connect unsplittable solutions to a relaxed version of the \emph{Ring Loading Problem}. To this end, consider the \emph{Ring Loading Problem} where demands are allowed to be routed splittably, i.e.\ a demand can be routed partly clockwise while the remaining part is routed counter-clockwise. The definition of the load of an edge and the load of a solution generalize naturally to the relaxed version. We denote with $L^*$ the optimum load of a split solution. The relaxed version of the \emph{Ring Loading Problem} has a linear programming formulation~\cite{Cosares1994} and can thus be solved in polynomial time. Further effort was put into finding more efficient algorithms (see~\cite{doi:10.1287/ijoc.8.3.235,doi:10.1287/opre.45.1.148,MR1612841,DELLAMICO1999119,MYUNG2004167,WANG200545}). It was also shown in~\cite{doi:10.1287/opre.45.1.148,MR1612841,DELLAMICO1999119} that $L \leq 2L^*$, and this bound is tight (\cite{doi:10.1287/opre.45.1.148,MR1612841}).

In a landmark paper, Schrijver, Seymour and Winkler~\cite{MR1612841} proved in this context that $L \leq L^* + \frac{3}{2}D$, where we denote with $D := \max_{i < j} d_{i,j}$ the maximum demand value. They furthermore gave the \enquote{guarantee} that $L \leq L^* + D$, which was later restated as conjecture in the survey on multicommodity flows by Shepherd~\cite{MR2513327}. More recently, Skutella~\cite{MR3463048} improved the upper bound by showing that $L \leq L^* + \frac{19}{14}D$. He also found an instance of the \emph{Ring Loading Problem} with $L = L^* + \frac{11}{10}D$, disproving the long-standing conjecture by Schrijver et al. and Shepherd. Skutella furthermore conjectured that $L \leq L^* + \frac{11}{10}D$. 

Interestingly, Schrijver et al.~\cite{MR1612841} gave an instance of the \emph{Ring Loading Problem} together with a split routing that cannot be turned into an unsplittable routing without increasing the load on some edge by at least $\frac{101}{100}D$, whereas Skutella~\cite{MR3463048} writes that this \enquote{does not imply a gap strictly larger than $D$ between the optimum values of split and unsplittable routings}. We show in~\cref{lem:boost} that this implication \emph{does} hold, and that Schrijver, Seymour and Winkler therefore (implicitly) found a counterexample to their own conjecture.

\paragraph{Our contributions}

The following theorem is the main contribution of this work.

\begin{theorem}
\label{thm:main}
Any split routing solution to the \emph{Ring Loading Problem} can be turned into an unsplittable routing while increasing the load on any edge by at most $\frac{13}{10}D$. In particular, we have $L \leq L^* + \frac{13}{10}D$.
\end{theorem}

In order to prove the theorem, we first define a general framework that unifies structural results of split routings introduced by Skutella~\cite{MR3463048}. We then apply this framework in a new way to obtain better upper bounds. This result is the first progress towards closing the remaining additive gap since Skutella~\cite{MR3463048}.

As all previous lower bound examples are of relative small size, it is interesting to settle these cases conclusively. We take a step into this direction by showing upper and lower bounds for small instances. The upper bounds are deduced from a mixed integer linear program that verifies for a given instance size that no worse examples can exist. Although the lower bounds also follow from this formulation, we provide further examples to enrich the view on instances where the difference $L - L^*$ is large with respect to $D$. In fact, we give an infinite family of instances with $L > L^* + D$.

A summary of previous results on lower and upper bounds together with new advancements is shown in~\cref{fig:bounds} on the right vertical line, while on the left results are given with respect to $\delta \in \left[0,\frac{1}{2}\right]$ that parametrizes instances of the \emph{Ring Loading Problem} and indicates whether a demand of medium size exists.

Just as Schrijver et al.~\cite{MR1612841} and Skutella~\cite{MR3463048} before, we mention a nice combinatorial implication of our result. Schrijver et al.~\cite{MR1612841} define $\beta$ to be the infimum of all reals $\alpha$ such that the following combinatorial statement holds: For all positive integers $m$ and nonnegative reals $u_1,\ldots,u_m$ and $v_1,\ldots,v_m$ with $u_i + v_i \leq 1$, there exist $z_1,\ldots,z_m$ such that for every k, $z_k \in \left\{v_k, -u_k\right\}$ and
\begin{equation*}
\left| \sum_{i = 1}^k z_i - \sum_{i = k +1}^m z_i \right| \leq \alpha.
\end{equation*}

Schrijver et al.~\cite{MR1612841} prove that $\beta \in \left[\frac{101}{100},\frac{3}{2}\right]$. Skutella~\cite{MR3463048} reduces the size of the interval to $\beta \in \left[\frac{11}{10}, \frac{19}{14}\right]$. As a result of our work, we obtain $\beta \in \left[\frac{11}{10}, \frac{13}{10}\right]$.

\paragraph{Further Related Work}

In the \emph{Ring Loading Problem} with integer demand splitting, each demand is allowed to be split into two integer parts which are routed in different directions along the ring. The objective is to find an integer split routing that minimizes the load. Let $L'$ be the load of an optimal integer split routing solution. Lee et al.~\cite{leearticle} showed an algorithm that returns an integer split routing solution with load at most $L' +1$. Schrijver et al.~\cite{MR1612841} found an optimal solution in pseudo-polynomial time. Vachani et al.~\cite{doi:10.1287/ijoc.8.3.235} provided an $O\left(n^3\right)$ algorithm. In~\cite{doi:10.1137/S0895480199358709} Myung presented an algorithm with runtime $O\left(nk\right)$ where $k$ is the number of non-zero demands. Wang~\cite{WANG200545} proved the existence of an $O\left(k + t_S\right)$ algorithm where $t_S$ is the time for sorting $k$ nodes.

More recently, the weighted \emph{Ring Loading Problem} was introduced where each edge has a weight associated with it, and the \emph{weighted load} of an edge is the product of its weight and the smallest integer greater or equal than its load. In the case where demand splitting is allowed, Nong et al.~\cite{Nong2009} gave an $O\left(n^2 k\right)$ algorithm. If integer demand splitting is allowed, the authors present a pseudo-polynomial time algorithm. Later, Nong et al.~\cite{NONG20102978} present an $O\left(n^3k\right)$ algorithm. If the demands have to be send unsplittably, Nong et al.~\cite{Nong2009} prove the existence of a PTAS.

In a broader context, the \emph{Ring Loading Problem} is a special case of unsplittable multicommodity flows. We mention the case of single source unsplittable flows, as similarities between theorems and conjectures for these problems exist (see~\cite{MR1722208,MR1888989,misc,10.1007/978-3-540-75520-3_36}). We also refer to the survey of Shepherd~\cite{MR2513327}.

\paragraph{Outline}

In~\cref{sec:prel}, we introduce some notation and provide useful results from Schrijver et al.~\cite{MR1612841} and Skutella~\cite{MR3463048} that we need. We then continue in~\cref{sec:upper} with the proof of~\cref{thm:main}. In~\cref{sec:bound}, we turn our attention to upper and lower bounds for small instances. We wrap everything up with our conclusions in~\cref{sec:concl}.

\section{Preliminaries}
\label{sec:prel}

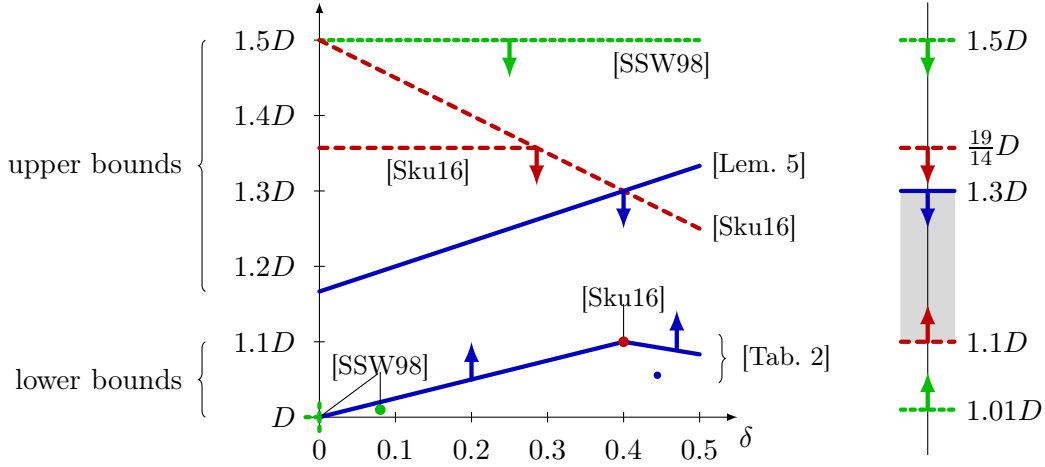
\begin{figure}
\centering
\usetikzlibrary{shapes.misc}
\begin{tikzpicture}[scale=1,
line/.style={draw, ultra thick, line cap = round},
ssw/.style={ssw_col, text = black, line, dotted},
martin/.style={martin_col, text = black, line, dashed},
new/.style={new_col, text = black,line, solid}
]

\definecolor{ssw_col}{RGB} {0,190,0}
\definecolor{martin_col}{RGB} {190,0,0}
\definecolor{new_col}{RGB} {0,0,190}

\def\radius{2pt}
\def\b{8cm}
\def\off{10pt}
\def\arr{.5cm}

\path[draw = white!85!black, fill = white!85!black, opacity = 1] (\b - \off, 11) to[] (\b -\off, 13) to[] (\b +\off, 13) to[] (\b +\off, 11) -- cycle;

\draw[-{latex}] (-.15, 10) to[] node[pos=1.02,below] {$\delta$} (5.5, 10);
\draw[-{latex}] (0, 9.85) to[] (0, 15.5);
\draw[] (\b, 9.5) to[] (\b, 15.5);

\foreach \x/\l in {0/0,1/0.1,2/0.2,3/0.3,4/0.4,5/0.5} {
\draw[] (\x cm, 10cm -2pt) to[] node[below=5pt] {$\l$} (\x cm, 10cm +2pt);
}

\foreach \y/\l in {0/D,1/1.1D,2/1.2D,3/1.3D,4/1.4D,5/1.5D} {
\draw[] (-2pt, 10cm + \y cm) to[] node[left=5pt] {$\l$} (+2pt, 10cm + \y cm);
}

\draw[decorate, decoration={brace}]  (-1.5cm,70/6) -- node[left=1ex] {upper bounds}  (-1.5cm,15);
\draw[decorate, decoration={brace}]  (-1.5cm,10) -- node[left=1ex] {lower bounds}  (-1.5cm,11);

\draw[decorate, decoration={brace}]  (5.25,11.1) -- node[black,right=1ex] {\small[Tab.~\ref{tab:it_res}]}  (5.25,190/18 - .1);

\draw[new] (0, 10/1) to[] (20/5,110/10);
\draw[new] (20/5, 110/10) to[] (10/2,130/12);
\draw[new, solid, -{latex}] (10/5, 105/10) to[] (10/5, 105cm/10 + \arr);
\draw[new, solid, -{latex}] (47/10, 109.166/10) to[] (47/10, 109.166cm/10 + \arr);

\draw[] (80/100, 1010/100) to[] node[black,above] {\small\cite{MR1612841}} +(0,.5cm);
\draw[shorten <= 1pt] (80/100, 1010/100) ++(0,.5cm) to[] (0,10);
\path[fill = ssw_col] (80/100, 1010/100) circle (\radius);

\draw[ssw] (0,15) to[] node[black,pos=.9, below] {\small\cite{MR1612841}} (5,15);
\draw[ssw, solid, -{latex}] (2.5, 15cm) to[] (2.5, 15 cm - \arr);
\draw[ssw] (0,10) to[] (-.25,10);
\draw[ssw] (0,10) to[] (0,10.25);
\draw[ssw] (0,10) to[] (0,9.8);

\draw[ssw] (\b -\off, 1010/100) to[] node[right = \off] {$1.01D$} (\b + \off, 1010/100);
\draw[ssw, solid, -{latex}] (\b, 101/10) to[] (\b, 101cm/10 + \arr);

\draw[ssw] (\b -\off,15) to[] node[right = \off] {$1.5D$} (\b +\off,15);
\draw[ssw, solid, -{latex}] (\b, 15cm) to[] (\b, 15 cm - \arr);

\draw[martin] (0, 15) to[] node[black,pos=1, right] {\small\cite{MR3463048}} (5,12.5);
\draw[martin] (0, 190/14) to[] node[black,pos=.5, below] {\small\cite{MR3463048}} (20/7,190/14);
\draw[martin, solid, -{latex}] (20/7, 190/14) to[] (20/7, 190cm/14 - \arr);

\draw[] (4, 11) to[] node[black,above] {\small\cite{MR3463048}} +(0cm,.5cm);
\path[fill = martin_col] (4, 11) circle (\radius);

\draw[martin] (\b -\off,190/14) to[] node[right = \off] {$\frac{19}{14}D$} (\b +\off,190/14);
\draw[martin, solid, -{latex}] (\b, 190/14) to[] (\b, 190cm/14 - \arr);
\draw[martin] (\b -\off,11) to[] node[right = \off] {$1.1D$} (\b +\off,11);
\draw[martin, solid, -{latex}] (\b, 11cm) to[] (\b, 11 cm + \arr);

\draw[new] (0, 70/6) to[] node[black,pos=1, right] {\small[Lem.~\ref{lem:upper}]} (5,40/3);

\coordinate[] (ce18) at (8/18 * 10, 19/18 * 10);

\path[fill = new_col] (ce18) circle(\radius - .6pt);

\draw[new] (\b -\off,13) to[] node[right = \off] {$1.3D$} (\b +\off,13);
\draw[new, -{latex}] (\b, 13) to[] (\b, 13cm - \arr);
\draw[new, -{latex}] (20/5, 13) to[] (20/5, 13cm - \arr);

\end{tikzpicture}
\caption{Summary of known results dependent on $\delta$ (left) and independent of $\delta$ (right). The currently best bounds are due to~\cref{thm:main} together with the lower bound in~\cite{MR3463048}.}
\label{fig:bounds}
\end{figure}

In this section, we introduce further notation and mention results already presented in~\cite{MR1612841,MR3463048}. We start with a preprocessing step to reduce the size and complexity of an instance to the \emph{Ring Loading Problem}.

Two demands $d_{i,j}$ and $d_{k,l}$ are \emph{parallel} if there exists a path from $i$ to $j$ and a path from $k$ to $l$ that are edge-disjoint, otherwise they are \emph{crossing}. Note that the demands $d_{i,j}$ and $d_{i,k}$ are parallel. 

As~\cref{thm:main} only argues about the load increase on all edges for split routing solutions, we can ignore and delete all demands that are routed unsplittably. The following observation shows that we can assume that there are not too many remaining demands.

\begin{observation}[\cite{MR1612841}]
\label{obs:parallel}
Given a split routing of two parallel demands $d_1$ and $d_2$. The routing can be altered such that at most one demand is routed splittably, without increasing the load on any edge.
\end{observation}

\begin{proof}
By the definition of parallel demands, we know that there are paths $P_i$, $i \in \left[2\right]$, connecting the nodes of demand $d_i$, such that $P_1 \cap P_2 = \emptyset$. Let $Q_i = E \setminus P_i$, $i \in \left[2\right]$. Note that $P_1$ is completely contained in $Q_2$ and vice versa. We denote with $x_i$ the amount of flow that demand $d_i$ is routing along $Q_i$, $i \in \left[2\right]$. If we now decrease the demand value routed along $Q_1$ and $Q_2$ by $\min \left\{ x_1,x_2\right\}$ and increase the demand value routed along $P_1$ and $P_2$ by the same amount, the load on the edges of $P_1$ and $P_2$ remain unchanged while the load on $Q_1 \cap Q_2$ decreases. Afterwards either $d_1$ or $d_2$ is routed unsplittably.
\end{proof}

If we apply~\cref{obs:parallel} and delete afterwards all demands that are routed unsplittably, we can concentrate on instances with pairwise crossing demands, implying in particular that every node is end point of at most one demand. If a node is not the end point of a demand, the load on its adjacent edges have the same value, allowing us to delete the node and merge the edges.

After this process we are left with a ring on $n = 2m$ nodes, demands $d_i := d_{i,i+m} > 0$ for $i \in \left[m\right]$ and a split routing. We denote for all $i \in \left[m\right]$ with $u_i > 0$ the amount of flow from demand $d_i$ routed clockwise and likewise with $v_i > 0$ the remainder of flow routed counter-clockwise. Note that $u_i + v_i = d_i$, $i \in \left[m\right]$. From now on we refer to an instance with this structure as \emph{split routing solution}. An example is given in~\cref{fig:prel} on the left.

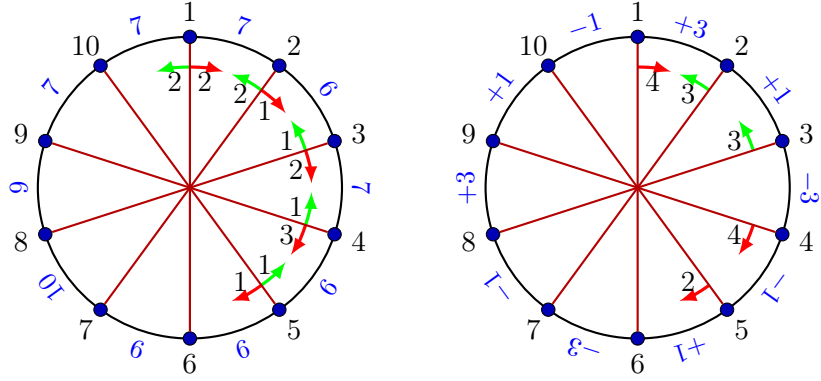
\begin{figure}
\centering
\def\radius{2cm}\def\delangle{16}
\begin{tikzpicture}[
]


\draw[thick] (0,0) circle (\radius);

\foreach \i [evaluate=\i as \j using int(\i +1)] in {0,1,...,9} {
	\node[mynode] (\i) at (90 - 360 / 10 * \i:\radius) {};
  \node[] () at (90 - 360 / 10 * \i:7/6*\radius) {\j};
}

\foreach \i/\l in {0/7,1/6,2/7,3/9,4/9,5/9,6/10,7/9,8/7,9/7} {
	\node[text = load_col, rotate = -360 / 20 - 360 / 10 * \i] at (90 - 360 / 20 - 360 / 10 * \i:9/8 * \radius) {\small$\l$};
}

\foreach \i/\u/\v in {0/2/2,1/1/2,2/2/1,3/3/1,4/1/1} {
	\draw[split_arc, split_col_u] (90 - 360 / 10 * \i:\radius - 1/5*\radius) arc[radius = \radius - 1/5*\radius, start angle = 90 - 360 / 10 * \i, delta angle = -\delangle];
	\draw[split_arc, split_col_v] (90 - 360 / 10 * \i:\radius - 1/5*\radius) arc[radius = \radius - 1/5*\radius, start angle = 90 - 360 / 10 * \i, delta angle = \delangle];

	\node[] at ({90 - 360 / 10 * \i - \delangle /2}:{\radius - 7/24*\radius}) {\u};
	\node[] at ({90 - 360 / 10 * \i + \delangle /2}:{\radius - 7/24*\radius}) {\v};
}

\foreach \i/\j/\d in {0/5/4,1/6/3,2/7/3,3/8/4,4/9/2} {
	\path[mydemand] (\i) to[] 
  (\j);
}



\end{tikzpicture}%
\qquad
\begin{tikzpicture}[
]


\draw[thick] (0,0) circle (\radius);

\foreach \i [evaluate=\i as \j using int(\i +1)] in {0,1,...,9} {
	\node[mynode] (\i) at (90 - 360 / 10 * \i:\radius) {};
  \node[] () at (90 - 360 / 10 * \i:7/6*\radius) {\j};
}

\foreach \i/\l in {0/+3,1/+1,2/-3,3/-1,4/+1,5/-3,6/-1,7/+3,8/+1,9/-1} {
	\node[text = load_col, rotate = -360 / 20 - 360 / 10 * \i] at (90 - 360 / 20 - 360 / 10 * \i:9/8 * \radius) {\small$\l$};
}

\foreach \i/\u/\v in {0/4/0,3/4/0,4/2/0} {
	\draw[split_arc, split_col_u] (90 - 360 / 10 * \i:\radius - 1/5*\radius) arc[radius = \radius - 1/5*\radius, start angle = 90 - 360 / 10 * \i, delta angle = -\delangle];

	\node[] at ({90 - 360 / 10 * \i - \delangle /2}:{\radius - 7/24*\radius}) {\u};
}

\foreach \i/\u/\v in {1/0/3,2/0/3} {
	\draw[split_arc, split_col_v] (90 - 360 / 10 * \i:\radius - 1/5*\radius) arc[radius = \radius - 1/5*\radius, start angle = 90 - 360 / 10 * \i, delta angle = \delangle];

	\node[] at ({90 - 360 / 10 * \i + \delangle /2}:{\radius - 7/24*\radius}) {\v};
}

\foreach \i/\j/\d in {0/5/4,1/6/3,2/7/3,3/8/4,4/9/2} {
	\path[mydemand] (\i) to[] 
  (\j);
}



\end{tikzpicture}%
\caption{An example of a split routing solution on $m = 5$ pairwise crossing demands with $u = \left(2,1,2,3,1\right)$ and $v = \left(2,2,1,1,1\right)$ together with the load on each edge (left). The corresponding unsplittable solution for $z = \left(v_1, -u_2, -u_3, v_4, v_5\right) = \left(2,-1,-2,1,1\right)$ together with load changes on every edge (right). The additive performance of $z$ is $3$.}
\label{fig:prel}
\end{figure}

The following definition describes for a given $\delta \in \left[0,\frac{1}{2}\right]$ all split routing solutions without demands of medium size (with respect to $\delta$) and ensures the existence of a demand on the boundary to medium demands. Formally we call these split routing solutions $\delta$-instances:

\begin{definition}
Let $\delta \in \left[0,\frac{1}{2}\right]$. We call a split routing solution a \emph{$\delta$-instance}, if for all $i \in \argmin_{j \in \left[m\right]}\left(\left|\frac{1}{2}D - d_j \right|\right)$ holds $d_i \in \left\{\delta D, \left(1 - \delta\right)D\right\}$.
\end{definition}

A $\frac{1}{2}$-instance for example has a demand of value $\frac{1}{2}D$, whereas a $0$-instance only has demands of value $D$. An important property of $\delta$-instances is that $d_i \in \left[0, \delta D\right] \cup \left[\left(1 - \delta\right)D, D\right]$ for all $i \in \left[m\right]$.

Any unsplittable solution has to decide for each demand $d_i$ whether $u_i$ units of flow are rerouted to use the counter-clockwise direction, or whether $v_i$ units of flow are rerouted to use the clockwise direction. We encode this decision using $z = \left(z_1,\ldots,z_m\right)$, with $z_i \in \left\{v_i, -u_i\right\}$ for all $i \in \left[m\right]$, where $z_i = v_i$ means that we send the demand completely in clockwise direction, whereas $z_i = -u_i$ means that we completely send the demand in counter-clockwise direction. In either case the $z_i$ values model exactly the increase of load on the clockwise edges from $i$ to $i+m$, and the decrease of load on the counter-clockwise edges. For $k \in \left[m\right]$ the load on an edge $\left\{k, k+1\right\}$ changes by
\begin{equation*}
\sum_{i = 1}^{k} z_i - \sum_{i = k+1}^m z_i,
\end{equation*}
while the load on the opposite edge $\left\{k + m, k + m +1\right\}$ changes by the negative amount. The maximum increase of load on any edge is therefore
\begin{equation*}
\max_{k \in \left[m\right]} \left| \sum_{i = 1}^k z_i - \sum_{i = k+1}^m z_i \right|.
\end{equation*}
As described by Skutella~\cite{MR3463048}, we refer to this quantity as the \emph{additive performance} of $z$. In~\cref{fig:prel} an example of the load change and the additive performance is given.

Let $x \in \RR$ be fixed, we define $p_z(k) := x + \sum_{i = 1}^k z_i$, for $k \in \left[m\right]$. We refer to $p_z$ as a \emph{pattern} starting at $x = p_z(0)$ and ending at $y = p_z(m)$. We denote with $a := \min_{k \in \left[m\right]} p_z(k)$ and $b := \max_{k \in \left[m\right]} p_z(k)$ the minimum and maximum of pattern $p_z$, respectively. We refer to $\left[a,b\right]$ as strip and say that the pattern $p_z$ lives on the strip $\left[a,b\right]$ of width $b - a$. As $p_z(k) - p_z(k -1) = z_i$, when we refer to a pattern $p_z$ we also refer to the corresponding unsplittable solution. As the choice of $x$ might vary, multiple patterns correspond to a single unsplittable solution. A pattern can be visualized as seen in~\cref{fig:prel_pattern}.

\begin{figure}
\centering
\begin{tikzpicture}[yscale=.5,xscale=1.25,
node style/.style={draw = black, fill = black, circle, inner sep = 0pt, minimum size = 3pt},
edge style/.style={draw},
line/.style={draw = black},
fwd/.style={draw,thick}]


\definecolor{a_col}{RGB}{0,190,0}
\definecolor{b_col}{RGB}{0,0,190}
\definecolor{c_col}{RGB}{190,0,0}

\path[fill=white!90!black] (0,1) rectangle (5,4);

\foreach \y in {2,...,3}
    \draw[help lines] (5,\y) -- (0,\y);

\draw[help lines] (0,1) -- (5,1) node[text = black, anchor = west] {$a = 1$};
\draw[help lines] (5,2) -- (0,2) node[text = black, anchor = east] {$x = 2$};
\draw[help lines] (0,3) -- (5,3) node[text = black, anchor = west] {$y = 3$};

\foreach \x in {0,...,5}
    \draw (\x,1pt) -- (\x,-3pt) node[anchor=north] {\x};

\path[line, {latex}-] (5.5,0) to[] node[anchor = north west,pos = 0] {$k$} (0,0) node[anchor=east] {$0$};
\path[line] (0,4) to[] (5,4) node[anchor=west] {$b = 4$};
\path[line, -{latex}] (0,0) to[] (0,4.5) node[anchor=east] {$p_z(k)$};

\node[node style] (0) at (0,2) {};
\node[node style] (1) at (1,4) {};
\node[node style] (2) at (2,3) {};
\node[node style] (3) at (3,1) {};
\node[node style] (4) at (4,2) {};
\node[node style] (5) at (5,3) {};

\path[edge style, fwd, a_col] (0) to[] (1);
\path[edge style, fwd, a_col] (1) to[] (2);
\path[edge style, fwd, a_col] (2) to[] (3);
\path[edge style, fwd, a_col] (3) to[] (4);
\path[edge style, fwd, a_col] (4) to[] (5);
\end{tikzpicture}
\caption{An example of a pattern $p_z$ that corresponds to the split routing given in~\cref{fig:prel} with $z = \left(2,-1,-2,1,1\right)$ and start point $x = 2$, end point $y = 3$, minimum value $a = 1$ and maximum value $b = 4$. The additive performance of the pattern due to~\cref{obs:obj} is $\max\left\{2b - x -y, x + y - 2a \right\} = 3$.}
\label{fig:prel_pattern}
\end{figure}
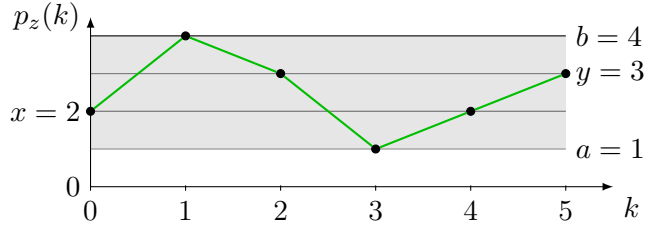

\begin{observation}[\cite{MR3463048}]
\label{obs:obj}
Given an unsplittable solution $z$ with corresponding pattern $p_z$ with start point $x$, end point $y$ living on a strip of $\left[a,b\right]$, then the additive performance of pattern $p_z$ is
\begin{equation}
\max_{k \in \left[m\right]} \left| \sum_{i = 1}^k z_i - \sum_{i = k +1}^m z_i \right| = \max\left\{ 2b - x - y, x + y - 2a\right\}.
\end{equation}
\end{observation}

\begin{proof}
By the definition of $p_z(k)$ we have
\begin{equation*}
\sum_{i = 1}^k z_i - \sum_{i = k +1}^m z_i = 2p_z(k) - x - y.
\end{equation*}
The claim then follows from the fact that the maximum in $\left| 2p_z(k) - x - y\right|$ is obtained at an index $k$ where $p_z(k)$ is either maximum or minimum.
\end{proof}

\begin{observation}[\cite{MR3463048}]
\label{obs:symm}
Let $\varepsilon > 0$. Given an unsplittable solution $z$ with corresponding pattern $p_z$ with start point $x$ and end point $y$ living on a strip of $\left[a,b\right]$, the additive performance of pattern $p_z$ is at most $b -a + \varepsilon$ if and only if the pattern starts at $x$ and ends at $y \in \left[y_{\text{opt}} - \varepsilon, y_{\text{opt}} + \varepsilon\right] \cap \left[a,b\right]$ with $y_{\text{opt}} := a + b - x$.
\end{observation}

\begin{proof}
The claim follows from the definitions together with~\cref{obs:obj}.
\end{proof}

Given a strip of width $D$, say $\left[0, D\right]$. Let $x \in \left[0,D\right]$, we denote throughout with $\bar{x} := D - x$ the reflection of $x$ across $\frac{1}{2}D$. We can construct a pattern with start point $x$ living on a strip $\left[a,b\right] \subseteq \left[0,D\right]$ by applying iteratively the following observation.

\begin{observation}
\label{obs:trivial}
Let $d = u + v$ with $u,v \geq 0$. If $I$ is an interval of size at least $d$ and $x \in I$, then $x + v \in I$ or $x - u \in I$ (or both).
\end{observation}

Formally, we construct the pattern $p_z$ by setting $p_z(0) = x$ for some $x \in \left[0,D\right]$ and choose $z_k \in \left\{-u_i, v_i\right\}$ iteratively such that $p_z(k) = p_z(k -1) + z_i \in \left[0,D\right]$ for all $k = 1,\ldots,m$, which always works by~\cref{obs:trivial}. If this decision is not unique, we set $z_k$ such that $\left|\frac{1}{2}D - p_z(k)\right|$ is minimal, i.e.\ $p_z(k)$ is as close as possible to the middle $\frac{1}{2}D$ of the interval $\left[0,D\right]$. Remaining ties are broken arbitrarily. A pattern that is constructed with respect to this procedure is called a \emph{forward greedy pattern}. For technical reasons, we call a forward greedy pattern $p_{z}$ \emph{proper} if its start point is far enough away from the boundary, i.e. $x \in \left[\frac{\delta}{4}D, \left(1 - \frac{\delta}{4}\right)D\right]$. This requirement is used in~\cref{lem:existence}.

We obtain a \emph{backward greedy pattern} $p_z$ by applying this procedure backwards. We define $p_z(m) = y$, for some $y \in \left[0,D\right]$, and iteratively choose $p_z(k -1) = p_z(k) - z_k$, for all $k = m,\ldots,1$, such that $\left| p_z(k -1) - \frac{1}{2}D\right|$ is minimal. We call a backward greedy pattern \emph{proper} if its end point is far enough away from the boundary, i.e. $y \in \left[\frac{\delta}{4}D,\left(1 - \frac{\delta}{4}\right)D\right]$.

A pattern is called a \emph{(proper) greedy pattern} if it is either a (proper) forward greedy pattern or a (proper) backward greedy pattern.

Using a forward greedy pattern starting at $\frac{1}{2}D$ together with~\cref{obs:symm}, Schrijver et al.~\cite{MR1612841} showed that any split routing solution to the \emph{Ring Loading Problem} can be turned into an unsplittable solution while increasing the load on any edge by at most $\frac{3}{2}D$.

Although the following structural properties of (greedy) patterns are crucial for our results, we refer the reader for complete proofs to~\cite{MR3463048}.

\begin{definition}[\cite{MR3463048}]
Let $\varepsilon \geq 0$. Two patterns $p_z$ and $p_{z'}$ are said to be $\varepsilon$-close if $|p_{z}(k) - p_{z'}(k)| \leq \varepsilon$ for some $k \in \left\{0,1,\ldots,m\right\}$. 
\end{definition}

The following lemma combines two $\varepsilon$-close patterns to a single pattern while preserving crucial properties.

\begin{lemma}[\cite{MR3463048}]
\label{lem:crossover}
Consider a fixed split routing solution. Let $p_{z'}$ be a pattern with start point $x'$ living on strip $\left[a',b'\right]$, and $p_{z''}$ a pattern with end point $y''$ living on strip $\left[a'', b''\right]$. If the two patterns are $\varepsilon$-close for some $\varepsilon \geq 0$, then there is a pattern $p_z$ living on a sub-strip of
\begin{equation*}
\left[\min\left\{a', a''\right\} - \frac{1}{2}\varepsilon, \max\left\{b', b''\right\} + \frac{1}{2} \varepsilon\right]
\end{equation*}
with start point $x$ and end point $y$ such that $x + y = x' + y''$.
\end{lemma}

This lemma describes situations where $\frac{\delta}{2}D$-close patterns exist.

\begin{lemma}[\cite{MR3463048}]
\label{lem:existence}
Consider three proper greedy patterns $p_{z_a}, p_{z_b}, p_{z_c}$, all three living on sub-strips of $\left[0,D\right]$. If the sorting of the patterns by their end points is not a cyclic permutation of the sorting of their start points, then (at least) two of the three patterns are $\frac{1}{2}\delta D$-close.
\end{lemma}


At the end of the section reconsider~\cref{fig:bounds}. Both previous and new results are shown with respect to $\delta$ on the left and the consequences for all instances independent of $\delta$ on the right.

\section{Improved Upper Bound}
\label{sec:upper}

In this section we prove~\cref{thm:main}. We start by defining a general framework that allows us to use~\cref{lem:existence,lem:crossover} in a very unified manner. The following definition is at the heart of this framework (see~\cref{fig:proof_def}).

\begin{definition}
\label{def:induced}
Given a greedy pattern $p_{z_a}$ living on a sub-strip of $\left[0,D\right]$ with start point $x_a$ and end point $y_a$, we call a forward greedy pattern $p_{z_b}$ \emph{induced} by $p_{z_a}$, if it lives on a sub-strip of $\left[0,D\right]$ with start point $x_b := \frac{2}{3}\bar{y}_a + \frac{1}{3}x_a$. Likewise, we call a backward greedy pattern $p_{z_c}$ \emph{induced} by $p_{z_a}$, if it lives on a sub-strip of $\left[0,D\right]$ with end point $y_c := \frac{2}{3}\bar{x}_a + \frac{1}{3}y_a$.
\end{definition}

If a greedy pattern $p_{z_a}$ and its induced patterns are proper, the following lemma ensures the existence of a pattern with an additive performance that only depends on the start and end points of $p_{z_a}$ together with $\delta$. It is therefore possible to pick a single pattern, check if its induced patterns are proper, and obtain a strong bound on the additive performance.

\begin{lemma}
\label{lem:induced_bound}
For a $\delta$-instance with $\delta \in \left[0,\frac{1}{2}\right]$, let $p_{z_a}$ be a greedy pattern living on a sub-strip of $\left[0,D\right]$ with start point $x_a$ and end point $y_a$. Denote with $p_{z_b}$ a forward greedy pattern induced by $p_{z_a}$ and with $p_{z_c}$ a backward greedy pattern induced by $p_{z_a}$. If all three greedy patterns are proper, then there exists a pattern with additive performance at most
\begin{equation*}
\max \left\{\frac{4}{3}D - \frac{1}{3}\left(x_a + y_a\right) + \frac{\delta}{2}D, \frac{2}{3}D + \frac{1}{3}\left(x_a + y_a\right) + \frac{\delta}{2}D\right\}.
\end{equation*}
\end{lemma}

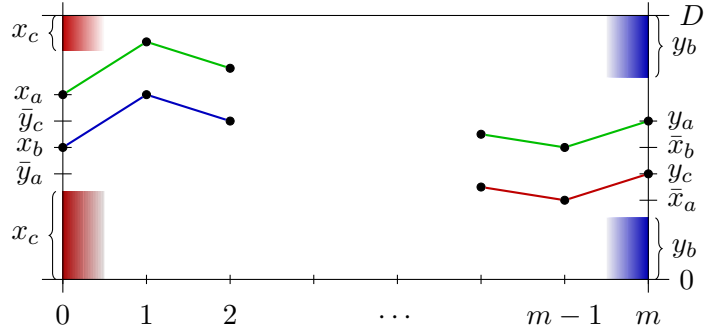
\begin{figure}
\centering
\begin{tikzpicture}[yscale=3.5,xscale=1.1,
node style/.style={draw = black, fill = black, circle, inner sep = 0pt, minimum size = 3pt},
edge style/.style={},
line/.style={draw = black},
fwd/.style={draw,thick},
bwd/.style={draw,thick},
opt/.style={draw=black, thick}
]

\path (-3,0) -- (10, 0); 

\definecolor{a_col}{RGB}{0,190,0}
\definecolor{b_col}{RGB}{0,0,190}
\definecolor{c_col}{RGB}{190,0,0}

\def\tick{.5pt}
\def\off{3pt}

\def\deltaa{1/3}

\pgfmathsetmacro\xa{.7}
\pgfmathsetmacro\xabar{1 - \xa}

\pgfmathsetmacro\ya{.6}
\pgfmathsetmacro\yabar{1 - \ya}

\pgfmathsetmacro\xb{2/3 * \yabar + 1/3 * \xa}
\pgfmathsetmacro\xbbar{1 - \xb}

\pgfmathsetmacro\yc{2/3 * \xabar + 1/3 * \ya}
\pgfmathsetmacro\ycbar{1 - \yc}

\pgfmathsetmacro\eps{1/3 * \xa + 1/3 * \ya -1/3 + \deltaa/2}

\pgfmathsetmacro\ybmin{0}
\pgfmathsetmacro\ybmax{\xbbar - \eps}
\pgfmathsetmacro\ybmintwo{\xbbar + \eps}
\pgfmathsetmacro\ybmaxtwo{1}

\pgfmathsetmacro\xcmin{0}
\pgfmathsetmacro\xcmax{\ycbar - \eps}
\pgfmathsetmacro\xcmintwo{\ycbar + \eps}
\pgfmathsetmacro\xcmaxtwo{1}

\pgfmathsetmacro\ua{.1}
\pgfmathsetmacro\va{.2}

\pgfmathsetmacro\ub{.1}
\pgfmathsetmacro\vb{.8}

\pgfmathsetmacro\uc{.05}
\pgfmathsetmacro\vc{.8}

\pgfmathsetmacro\vd{(\ya - \yc)/2}
\pgfmathsetmacro\ud{1 - \deltaa - \vd}

\foreach \i/\l in {0/0,1/1,2/2,3/,4/\cdots,5/,6/m-1,7/m} {
\draw (\i, -\tick) to[] node[text height = 0pt, below = 13pt] {$\l$} (\i, \tick);
}

\fill[left color=white, right color = b_col] (7,\ybmin cm) -- (7,\ybmax cm) -- (6.5,\ybmax cm) -- (6.5,\ybmin cm);
\fill[left color=white, right color = b_col] (7,\ybmintwo cm) -- (7,\ybmaxtwo cm) -- (6.5,\ybmaxtwo cm) -- (6.5,\ybmintwo cm);
\fill[left color=c_col, right color = white] (0,\xcmin cm) -- (0,\xcmax) -- (.5,\xcmax) -- (.5,\xcmin cm);
\fill[left color=c_col, right color = white] (0,\xcmintwo cm) -- (0,\xcmaxtwo cm) -- (.5,\xcmaxtwo cm) -- (.5,\xcmintwo cm);

\path[line] (7.25,0) to[] node[pos = 0, right] {$0$} (-.25,0);
\path[line] (7.25,1) to[]  node[pos = 0, right] {$D$} (-.25,1);
\path[line] (0,-.05) to[] (0, 1.05);
\path[line] (7,-.05) to[] (7,1.05);

\draw (-\off,\xa cm) to[] node[pos = 0, left] {$x_a$} (\off,\xa cm);
\draw (-\off,\xb cm) to[] node[pos = 0, left] {$x_b$} (\off,\xb cm);
\draw (-\off,\ycbar cm) to[] node[pos = 0, left] {$\bar{y}_c$} (\off,\ycbar cm);
\draw (-\off,\yabar cm) to[] node[pos = 0, left] {$\bar{y}_a$} (\off,\yabar cm);

\draw (7cm -\off,\xabar cm) to[] node[pos = 1, right] {$\bar{x}_a$} (7cm +\off,\xabar cm);
\draw (7cm-\off,\xbbar cm) to[] node[pos = 1, right] {$\bar{x}_b$} (7cm +\off,\xbbar cm);
\draw (7cm -\off,\yc cm) to[] node[pos = 1, right] {$y_c$} (7cm +\off,\yc cm);
\draw (7cm -\off,\ya cm) to[] node[pos = 1, right] {$y_a$} (7cm +\off,\ya cm);

\draw[decorate, decoration={brace}, xshift = -.5ex]  (0,\xcmin cm) -- node[left=.5ex] {$x_c$}  (0,\xcmax);
\draw[decorate, decoration={brace}, xshift = -.5ex]  (0,\xcmintwo) -- node[left=.5ex] {$x_c$}  (0,\xcmaxtwo cm);

\draw[decorate, decoration={brace, aspect=.5}, xshift=.51ex]  (7,\ybmax cm) -- node[pos=.5, right=.5ex] {$y_b$}  (7,\ybmin cm);
\draw[decorate, decoration={brace, aspect=.5}, xshift=.51ex]  (7,\ybmaxtwo cm) -- node[pos=.5, right=.5ex] {$y_b$}  (7,\ybmintwo);

\node[node style] (010fwd) at (7,\ya cm) {};
\node[node style] (100fwd) at (6,\ya cm - \vd cm) {};
\node[node style] (200fwd) at (5,\ya cm - \vd cm + \uc cm) {};

\path[edge style,bwd,a_col] (010fwd) to[] (100fwd);
\path[edge style,bwd,a_col] (100fwd) to[] (200fwd);

\node[node style] (002bwd) at (0,\xa cm) {};
\node[node style] (102bwd) at (1,\xa cm + \va cm) {};
\node[node style] (202bwd) at (2,\xa cm + \va cm - \ub cm) {};

\path[edge style,fwd,a_col] (002bwd) to[] (102bwd);
\path[edge style,fwd,a_col] (102bwd) to[] (202bwd);

\node[node style] (821fwd) at (5,\yc cm - \vd cm + \uc cm) {};
\node[node style] (921fwd) at (6,\yc cm - \vd cm) {};
\node[node style] (101fwd) at (7,\yc cm) {};

\path[edge style,bwd,c_col] (821fwd) to[] (921fwd);
\path[edge style,bwd,c_col] (921fwd) to[] (101fwd);

\node[node style] (002fwd) at (0,\xb cm) {};
\node[node style] (102fwd) at (1,\xb cm + \va cm) {};
\node[node style] (202fwd) at (2,\xb cm + \va cm - \ub cm) {};

\path[edge style,fwd,b_col] (002fwd) to[] (102fwd);
\path[edge style,fwd,b_col] (102fwd) to[] (202fwd);
\end{tikzpicture}
\caption{An illustration of a greedy pattern $p_{z_a}$ with induced forward greedy pattern $p_{z_b}$ and induced backward greedy pattern $p_{z_c}$ as defined in~\cref{def:induced}.}
\label{fig:proof_def}
\end{figure}

\begin{proof}
We first show that the sorting of the start points is not a cyclic permutation of the sorting of the end points. This allows us to use~\cref{lem:existence} that guarantees the existence of two patterns that are $\frac{\delta}{2}D$-close. We then conclude the lemma by showing that if any two of the three patterns are $\frac{\delta}{2}D$-close, that there exists a pattern with the required additive performance. In~\cref{fig:proof_def} is an illustration of the procedure.

By the definition of $p_{z_b}$ as forward greedy pattern induced by $p_{z_a}$ and $p_{z_c}$ as backward greedy pattern induced by $p_{z_a}$, we know that $x_b = \frac{2}{3}\bar{y}_a + \frac{1}{3}x_a$ and $y_c = \frac{2}{3}\bar{x}_a + \frac{1}{3}y_a$. The definitions are such that the interval between $x_a$ and $\bar{y}_a$ is divided into three equal parts by the points $x_b$ and $\bar{y}_c$. A straightforward computation shows that
\begin{equation}
\label{eq:mid_point}
\bar{y}_c = \frac{x_a + x_b}{2},\qquad\bar{x}_b = \frac{y_a + y_c}{2},
\end{equation}
i.e. the optimal start point for pattern $p_{z_c}$ is in the middle between $x_a$ and $x_b$, and the optimal end point of $p_{x_b}$ is in the middle between $y_a$ and $y_c$. Because $\left|x_a - \bar{y}_a\right| = \left|\bar{x}_a - y_a\right|$, and the symmetric definitions of $x_b$ and $y_c$, it also holds that $\left|x_a - x_b\right| = \left|y_a - y_c\right|$.

For the sake of brevity we define $\varepsilon := \max\left\{\frac{1}{3}D - \frac{1}{3}\left(x_a + y_a\right) + \frac{\delta}{2}D, \frac{1}{3}\left(x_a + y_a\right) - \frac{1}{3}D + \frac{\delta}{2}D\right\}$. In fact, we want to show that there exists a pattern with additive performance at most $D + \varepsilon$.

We now show that $\left| x_a - x_b \right|\leq 2\varepsilon$, which also implies that $\left|y_a - y_c\right| \leq 2\varepsilon$. By definition of $x_b = \frac{2}{3}\bar{y}_a + \frac{1}{3}x_a = \frac{2}{3}D - \frac{2}{3}y_a + \frac{1}{3}x_a$, we can conclude that
\begin{equation}
\label{eq:dist}
\begin{split}
\left| x_a - x_b \right| &= \max \left\{ x_b - x_a , x_a - x_b \right\} = \max \left\{ \frac{2}{3}D - \frac{2}{3}\left(x_a + y_a\right), \frac{2}{3}\left(x_a + y_a\right) - \frac{2}{3}D\right\} \\
& \leq 2 \max \left\{ \frac{1}{3}D - \frac{1}{3}\left(x_a + y_a\right) + \frac{\delta}{2}D, \frac{1}{3}\left(x_a + y_a\right) - \frac{1}{3}D + \frac{\delta}{2}D \right\} = 2 \varepsilon,
\end{split}
\end{equation}
where the inequality follows from the fact that $\delta \geq 0$. With~\cref{eq:mid_point,eq:dist}, it follows that
\begin{equation}
\label{eq:dist_2}
\left|\bar{y}_c - x_b \right| = \left|\bar{y}_c - x_a\right| = \left|\bar{x}_b - y_a \right| = \left|\bar{x}_b - y_c\right| \leq \varepsilon.
\end{equation}

If for the start point of $p_{z_c}$ holds $x_c \in \left[\bar{y}_c - \varepsilon, \bar{y}_c + \varepsilon\right]$, we know by~\cref{obs:symm} that the additive performance of $p_{z_c}$ is at most $D + \varepsilon$, from which the lemma follows. We can thus assume that $x_c \in \left[0, \bar{y}_c - \varepsilon\right] \cup \left[\bar{y}_c + \varepsilon, D\right]$. Using~\cref{eq:dist_2}, we can therefore conclude that either $x_c \leq \bar{y}_c -\varepsilon \leq \min \left\{ x_a, x_b\right\}$ or $\max \left\{ x_a, x_b \right\} \leq \bar{y}_c + \varepsilon \leq x_c$.

Equivalently, if for the end point of $p_{z_b}$ holds $y_b \in \left[\bar{x}_b - \varepsilon, \bar{x}_b + \varepsilon\right]$, we know by~\cref{obs:symm} that the additive performance of $p_{z_b}$ is at most $D + \varepsilon$, from which the lemma follows. We can thus assume that $y_b \in \left[0, \bar{x}_b - \varepsilon\right] \cup \left[\bar{x}_b + \varepsilon, D\right]$. Using~\cref{eq:dist_2}, we can therefore conclude that either $y_b \leq \bar{x}_b -\varepsilon \leq \min \left\{ y_a, y_c\right\}$ or $\max \left\{ y_a, y_c \right\} \leq \bar{x}_b + \varepsilon \leq y_b$.

Assume first that $\bar{y}_a \leq x_a$ (as shown in~\cref{fig:proof_def}), then $\bar{y}_a \leq x_b \leq \bar{y}_c \leq x_a$ and $\bar{x}_a \leq y_c \leq \bar{x}_b \leq y_a$, which implies that either $x_c \leq x_b \leq x_a$ or $x_b \leq x_a \leq x_c$ and that either $y_b \leq y_c \leq y_a$ or $y_c \leq y_a \leq y_b$. In either case, the sorting of the patterns by their start points is not a cyclic permutation of the patterns by their end points.

Assume now that $x_a \leq \bar{y}_a$, then $x_a \leq \bar{y}_c \leq x_b \leq \bar{y}_a$ and $y_a \leq \bar{x}_b \leq y_c \leq \bar{x}_a$, which implies that either $x_c \leq x_a \leq x_b$ or $x_a \leq x_b \leq x_c$ and that either $y_b \leq y_a \leq y_c$ or $y_a \leq y_c \leq y_b$. In either case, the sorting of the patterns by their start points is not a cyclic permutation of the patterns by their end points.

As $p_{z_a}$, $p_{z_b}$ and $p_{z_c}$ are proper greedy patterns, we can apply~\cref{lem:existence}, ensuring the existence of two patterns that are $\frac{\delta}{2}D$-close. We conclude the proof by showing that the closeness of any two patterns guarantees the existence of a pattern with the claimed additive performance.

\begin{enumerate}[i)]
\item Assume that $p_{z_a}$ and $p_{z_b}$ are $\frac{\delta}{2}D$-close. Then Lemma~\ref{lem:crossover} assures the existence of a pattern with start point $x$ and end point $y$ such that $x + y = x_b + y_a = \frac{2}{3}D + \frac{1}{3}\left(x_a + y_a\right)$ on a sub-strip of $\left[-\frac{\delta}{4}D,D + \frac{\delta}{4}D\right]$. Using~\cref{obs:obj}, a straightforward calculation shows that this pattern has additive performance at most
\begin{equation*}
\max \left\{ \frac{4}{3}D - \frac{1}{3}\left(x_a + y_a\right) + \frac{\delta}{2}D, \frac{2}{3}D + \frac{1}{3}\left(x_a + y_a\right) + \frac{\delta}{2}D\right\}.
\end{equation*}

\item Assume that $p_{z_a}$ and $p_{z_c}$ are $\frac{1}{2}\delta D$-close. Then Lemma~\ref{lem:crossover} assures the existence of a pattern with start point $x$ and end point $y$ such that $x + y = x_a + y_c = \frac{2}{3}D + \frac{1}{3}\left(x_a + y_a\right)$ on a sub-strip of $\left[-\frac{\delta}{4}D,D + \frac{\delta}{4}D\right]$. Using~\cref{obs:obj}, a straightforward calculation shows that this pattern has additive performance at most
\begin{equation*}
\max \left\{ \frac{4}{3}D - \frac{1}{3}\left(x_a + y_a\right) + \frac{\delta}{2}D, \frac{2}{3}D + \frac{1}{3}\left(x_a + y_a\right) + \frac{\delta}{2}D\right\}.
\end{equation*}

\item Assume that $p_{z_b}$ and $p_{z_c}$ are $\frac{1}{2}\delta D$-close. Then Lemma~\ref{lem:crossover} assures that there exists a pattern with start point $x$ and end point $y$ such that $x + y = x_b + y_c = \frac{4}{3}D - \frac{1}{3}\left(x_a + y_a\right)$ on a sub-strip of $\left[-\frac{\delta}{4}D,D + \frac{\delta}{4}D\right]$. Using~\cref{obs:obj}, a straightforward calculation shows that this pattern has additive performance at most
\begin{equation*}
\max \left\{ \frac{2}{3}D + \frac{1}{3}\left(x_a + y_a\right) + \frac{\delta}{2}D, \frac{4}{3}D - \frac{1}{3}\left(x_a + y_a\right) + \frac{\delta}{2}D\right\}.
\end{equation*}
\end{enumerate}
In either case, the lemma follows.
\end{proof}

If both the start point and the end point of a greedy pattern $p_{z_a}$ are far enough away from the boundary, the next lemma ensures that its induced patterns are proper. Note that this is a stronger requirement on $p_{z_a}$ than being a proper greedy pattern.

\begin{lemma}
\label{lem:strong_induced}
For a $\delta$-instance with $\delta \in \left[0,\frac{1}{2}\right]$, let $p_{z_a}$ be a greedy pattern living on a sub-strip of $\left[0,D\right]$ with start point $x_a$ and end point $y_a$. Denote with $p_{z_b}$ a forward greedy pattern induced by $p_{z_a}$ and with $p_{z_c}$ a backward greedy pattern induced by $p_{z_a}$. If $x_a,y_a \in \left[\frac{\delta}{4}D, \left(1 - \frac{\delta}{4}\right)D\right]$, then $p_{z_a}$, $p_{z_b}$ and $p_{z_c}$ are proper.
\end{lemma}

\begin{proof}
By definition, $p_{z_a}$ is proper. For $x \in \left[0,D\right]$ holds that $\bar{x} \in \left[\frac{\delta}{4}D, \left(1 - \frac{\delta}{4}\right)D\right]$ if and only if $x \in \left[\frac{\delta}{4}D, \left(1 - \frac{\delta}{4}\right)D\right]$. By assumption, we therefore also know that $\bar{x}_a$ and $\bar{y}_a$ are far enough away from the boundary. By the definition of induced patterns, we know that $x_b = \frac{2}{3}\bar{y}_a + \frac{1}{3}x_a$, which implies in particular that $\min \{x_a, \bar{y}_a\} \leq x_b \leq \max \{x_a, \bar{y}_a \}$. The start point $x_b$ of the induced forward greedy pattern is consequently far enough away from the boundary. Using the same argumentation for the definition of $y_c = \frac{2}{3}\bar{x}_a + \frac{1}{3}y_a$, the lemma follows.
\end{proof}

A crucial part of the proof of~\cref{thm:main}, and the main contribution of this work is the following auxiliary lemma.

\begin{lemma}
\label{lem:upper}
For a $\delta$-instance with $\delta \in \left[0, \frac{1}{2}\right]$ there exists a pattern with additive performance at most $\left(\frac{7}{6} + \frac{\delta}{3}\right) D$.
\end{lemma}

\begin{proof}
We start the proof by modifying the instance such that the special demand of value either $\delta D$ or $\left(1 - \delta\right)D$ is the last demand. These two cases will be treated separately. In either case, we then use the nice structure of the newly created instance to find a greedy pattern that can be used with Lemma~\ref{lem:induced_bound}.

Let $d_i$ be the demand that minimizes $\left|\frac{1}{2}D - d_i\right|$ over all $i \in \left[m\right]$. By the definition of a $\delta$-instance, we know that $d_i \in \left\{\delta D, \left(1 - \delta\right)D\right\}$.

We now rotate the instance such that the specially chosen demand $d_i$ has index $m$, and is thus the last demand of the instance. By a slight abuse of notation we will refer to this newly created instance again as instance. Recall that now the demand $d_m$ has the property that $d_m \in \left\{ \delta D, \left(1 - \delta\right)D\right\}$.

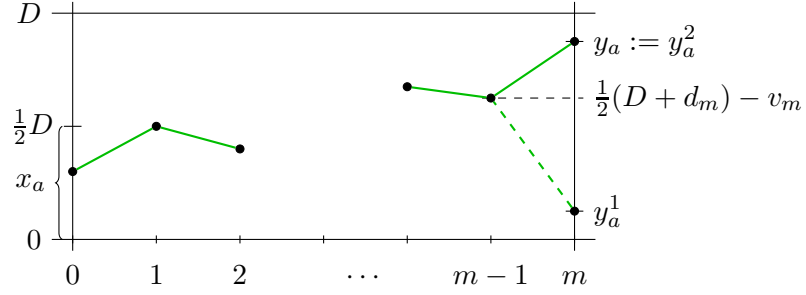
\begin{figure}
\centering
\begin{tikzpicture}[yscale=3,xscale=1.1,
node style/.style={draw = black, fill = black, circle, inner sep = 0pt, minimum size = 3pt},
edge style/.style={},
line/.style={draw = black},
fwd/.style={draw,thick},
bwd/.style={draw,thick},
opt/.style={draw=black, thick}
]

\path (-3,0) -- (9, 0); 

\definecolor{a_col}{RGB}{0,190,0}
\definecolor{b_col}{RGB}{0,0,190}
\definecolor{c_col}{RGB}{190,0,0}

\def\tick{.5pt}
\def\off{3pt}

\def\deltaa{1/4}
\pgfmathsetmacro\gammaa{1/6 + \deltaa/3}

\pgfmathsetmacro\xa{.3}

\pgfmathsetmacro\ya{1 - \deltaa/2}
\pgfmathsetmacro\yaopp{1 - \ya}

\pgfmathsetmacro\ua{.1}
\pgfmathsetmacro\va{.2}

\pgfmathsetmacro\ub{.1}
\pgfmathsetmacro\vb{.8}

\pgfmathsetmacro\uc{.05}
\pgfmathsetmacro\vc{.8}

\pgfmathsetmacro\vd{\ya/3.5}
\pgfmathsetmacro\ud{1 - \deltaa - \vd}

\foreach \i/\l in {0/0,1/1,2/2,3/,4/,5/m-1,6/m} {
\draw (\i, -\tick) to[] node[text height = 0pt, below = 13pt] {$\l$} (\i, \tick);
}
\path (3.5, -\tick) to[] node[text height = 0pt, below = 13pt] {$\cdots$} (3.5, \tick);

\draw[decorate, decoration={brace}, xshift = -.75ex]  (0,0 cm) -- node[left=.5ex] {$x_a$}  (0,.5);

\path[line] (6.25,0) to[] node[pos = 1, left] {$0$} (-.25,0);
\path[line] (6.25,1) to[]  node[pos = 1, left] {$D$} (-.25,1);
\path[line] (0,-.05) to[] (0, 1.05);
\path[line] (6,-.05) to[] (6,1.05);

\draw (-\off,.5 cm) to[] node[pos = 0, left] {$\frac{1}{2}D$} (\off,.5 cm);

\draw (6cm -\off, \ya cm) to[] node[pos = 1, right] {$y_a := y_a^2$} (6cm + \off,\ya cm);
\draw[dashed] (5cm - \off,\ya cm - \vd cm) to[] node[pos = 1, right] {$\frac{1}{2}(D + d_m) - v_m$} (6cm + \off,\ya cm - \vd cm);
\draw (6cm - \off,\yaopp cm) to[] node[pos = 1, right] {$y_a^1$} (6cm + \off,\yaopp cm);

\node[node style] (000fwd) at (6,\yaopp cm) {};
\node[node style] (010fwd) at (6,\ya cm) {};
\node[node style] (100fwd) at (5,\ya cm - \vd cm) {};
\node[node style] (200fwd) at (4,\ya cm - \vd cm + \uc cm) {};

\path[edge style,bwd,a_col, dashed] (000fwd) to[] (100fwd);
\path[edge style,bwd,a_col] (010fwd) to[] (100fwd);
\path[edge style,bwd,a_col] (100fwd) to[] (200fwd);

\node[node style] (002bwd) at (0,\xa cm) {};
\node[node style] (102bwd) at (1,\xa cm + \va cm) {};
\node[node style] (202bwd) at (2,\xa cm + \va cm - \ub cm) {};

\path[edge style,fwd,a_col] (002bwd) to[] (102bwd);
\path[edge style,fwd,a_col] (102bwd) to[] (202bwd);
\end{tikzpicture}
\caption{An example of the construction step in~\cref{lem:upper}. If $x_a \leq \frac{1}{2}D$, we extend the pattern with $y_a = \frac{1}{2}(D + d_m)$.}
\label{fig:proof_ex}
\end{figure}

The following procedure is similar to the one described by Skutella~\cite{MR3463048} when dealing with demands of medium size. An example is depicted in Figure~\ref{fig:proof_ex}. We first delete the last demand $m$ to obtain a smaller instance. We define a backward greedy pattern ending at $\frac{1}{2}\left(D + d_m\right) - v_m$ and starting at some $x_a \in \left[0,D\right]$. This backward greedy pattern can be extended in two possible ways to create a pattern that includes demand $m$, once with end point $y_a^1 := \frac{1}{2}\left(D - d_m\right)$ and once with end point $y_a^2 := \frac{1}{2}\left(D + d_m\right)$. A crucial observation is that both possible extensions produce a valid backward greedy pattern for the original instance. Depending on the particular start point $x_a$, we choose in which way the pattern will be extended: If $x_a \leq \frac{1}{2}D$, we extend the pattern with end point $y_a^2$, otherwise we extend the pattern with $y_a^1$. For the rest of the proof, we may assume that $x_a$ is at most $\frac{1}{2}D$ and the pattern is therefore extended with end point $y_a := y_a^2$. This assumption can be made, as the following construction is highly symmetric with respect to $y_a^1$ and $y_a^2$, in fact, all arguments remain valid if we change start and end points of subsequent patterns by reflecting their value around $\frac{1}{2}D$. Let $p_{z_a}$ denote the resulting backward greedy pattern starting at $x_a$ and ending at $y_a$.

We consider two cases, first that $d_m = \left(1 - \delta\right)D$ and second that $d_m = \delta D$. For the sake of brevity, we define $\varepsilon := \frac{1}{6}D + \frac{\delta}{3}D$. In fact, we want to find a pattern with additive performance at most $D + \varepsilon$.

\paragraph{Case a)} If $d_m = \left(1 -\delta\right)D$, we can rewrite $y_a = \frac{1}{2}\left(D + d_m\right) = D - \frac{\delta}{2}D$. The lemma follows immediately from~\cref{obs:symm} if $x_a \in \left[\bar{y}_a - \varepsilon, \bar{y}_a + \varepsilon\right]$. We can therefore assume that $x_a$ falls either into the interval $\left[0, \frac{\delta}{6}D - \frac{1}{6}D\right]$ or into the interval $\left[\frac{1}{6}D + \frac{5}{6}\delta D, \frac{1}{2}D\right]$. Recall the assumption that $x_a$ is at most $\frac{1}{2}D$. It is easy to see that $\frac{\delta}{6}D - \frac{1}{6}D$ is negative for all $\delta \in \left[0,\frac{1}{2}\right]$. It follows that $x_a \in \left[\frac{1}{6}D + \frac{5}{6}\delta D, \frac{1}{2}D\right]$. Note that this interval is also empty for all $\delta > \frac{2}{5}$, and the lemma is trivially correct. In fact, this is exactly the argumentation used by Skutella~\cite{MR3463048} in his proof of Lemma~\ref{lem:small}.

As $y_a = D - \frac{\delta}{2}D \in \left[\frac{\delta}{4}D, \left(1 - \frac{\delta}{4}\right)D\right]$, the backward greedy pattern $p_{z_a}$ is proper. Because $\frac{1}{2}D \geq x_a \geq \frac{1}{6}D + \frac{5}{6}\delta D \geq \frac{\delta}{4}D$, it furthermore holds that $x_a$ is far enough away from the boundary. We can thus apply~\cref{lem:strong_induced} together with~\cref{lem:induced_bound} and the fact that $x_a + y_a \in \left[ \frac{7}{6}D + \frac{\delta}{3}D, \frac{3}{2}D - \frac{\delta}{2}D\right]$ to obtain a pattern with additive performance at most
\begin{equation*}
\max \left\{ \frac{17}{18}D + \frac{7}{18} \delta D, \frac{7}{6}D + \frac{\delta}{3}D \right\} = \left(\frac{7}{6} + \frac{\delta}{3}\right) D.
\end{equation*} 

\paragraph{Case b)} If $d_m = \delta D$, we can rewrite $y_a = \frac{1}{2}\left(D + d_m\right) = \frac{1}{2}D + \frac{\delta}{2}D$. The lemma follows immediately from~\cref{obs:symm} if $x_a \in \left[\bar{y}_a - \varepsilon, \bar{y}_a + \varepsilon \right]$. We can therefore assume that $x_a$ falls either into the interval $\left[0, \frac{1}{3}D - \frac{5}{6}\delta D\right]$ or into the interval $\left[\frac{2}{3}D - \frac{\delta}{6} D, \frac{1}{2}D\right]$. Recall the assumption that $x_a$ is at most $\frac{1}{2}D$. It is easy to see that $\frac{2}{3}D - \frac{\delta}{6}D \geq \frac{1}{2}D$ for all $\delta \in \left[0,\frac{1}{2}\right]$. It follows that $x_a \in \left[0, \frac{1}{3}D - \frac{5}{6}\delta D\right]$. Note that this interval is also empty for all $\delta > \frac{2}{5}$, and the lemma is trivially correct. We need this assumption, when arguing that we can apply Lemma~\ref{lem:induced_bound}.

As $y_a = \frac{1}{2}D + \frac{\delta}{2}D \in \left[\frac{\delta}{4}D, \left(1 - \frac{\delta}{4}\right)D\right]$, the backward greedy pattern $p_{z_a}$ is proper. As $x_a$ might be zero, we cannot apply~\cref{lem:strong_induced}. We therefore have to argue that the induced patterns are proper. Let $p_{z_b}$ be a forward greedy pattern induced by $p_{z_a}$ and $p_{z_c}$ be a backward greedy pattern induced by $p_{z_a}$. By definition, we have $x_b = \frac{2}{3}\bar{y}_a + \frac{1}{3}x_a$ and $y_c = \frac{2}{3}\bar{x}_a + \frac{1}{3}y_a$. By substituting the definitions and bounds of $y_a$ and $x_a$, we obtain
\begin{equation*}
x_b = \frac{1}{3}D - \frac{\delta}{3}D + \frac{1}{3}x_a \geq \frac{1}{3}D - \frac{\delta}{3}D \geq \frac{\delta}{4}D.
\end{equation*}
The start point $x_b$ is therefore far enough away from the boundary and the pattern $p_{z_b}$ is thus proper. We similarly obtain
\begin{equation*}
y_c = \frac{5}{6}D + \frac{\delta}{6}D - \frac{2}{3}x_a \leq \frac{5}{6}D + \frac{\delta}{6}D \leq D - \frac{\delta}{4}D,
\end{equation*}
for all $\delta \in \left[0,\frac{2}{5}\right]$. As we assumed that $\delta \leq \frac{2}{5}$, the backward greedy pattern $p_{z_c}$ induced by $p_{z_a}$ is proper. We can thus apply~\cref{lem:induced_bound} together with the fact that $x_a + y_a \in \left[ \frac{1}{2}D + \frac{\delta}{2}D, \frac{5}{6}D - \frac{\delta}{3}D\right]$ to obtain a pattern with additive performance at most
\begin{equation*}
\max \left\{ \frac{7}{6}D + \frac{\delta}{3}D, \frac{17}{18}D + \frac{7}{18} \delta D\right\} = \left(\frac{7}{6} + \frac{\delta}{3}\right) D.
\end{equation*} 
In either case, the lemma follows.
\end{proof}

An easy consequence of Lemma~\ref{lem:upper} is that there exists for any split routing solution a pattern with additive performance at most $\frac{4}{3}D$, which already improves upon the best known previous result of $\frac{19}{14}D$ from Skutella~\cite{MR3463048}. However, when combined with Skutellas~\cite{MR3463048} result on instances with medium demands (see~\cref{lem:small}), we obtain our main Theorem~\ref{thm:main}.

\begin{lemma}[\cite{MR3463048}]
\label{lem:small}
For any $\delta$-instance with $\delta \in \left[0,\frac{1}{2}\right]$ there exists a pattern with additive performance at most $\left(\frac{3}{2} - \frac{\delta}{2}\right)D$.
\end{lemma}

\begin{proof}[Proof of Theorem~\ref{thm:main}]
Let $\delta \in \left[0, \frac{1}{2}\right]$ be such that the given split routing solution is a $\delta$-instance. If $\delta \geq \frac{2}{5}$, the theorem follows from \cref{lem:small}, as $\left( \frac{3}{2} - \frac{\delta}{2}\right)D \leq 1.3D$. Otherwise, the theorem follows from \cref{lem:upper}, as $\left( \frac{7}{6} + \frac{\delta}{3} \right)D \leq 1.3D$.
\end{proof}

\section{Bounds for Small Instances}
\label{sec:bound}

In this section, we show lower and upper bounds for small instances of the \emph{Ring Loading Problem}. This is of interest, as strong lower bound examples seem to exist for fairly small values of $m$ (see~\cref{sec:it_lb}). Therefore, dealing with these cases conclusively might lead to a deeper understanding on the correct bounds.

The main result of this section are bounds on the maximal load increase while turning a split routing solution into an unsplittable solution for instances of the \emph{Ring Loading Problem} where at most $7$ demands are routed splittably.

\begin{theorem}
\label{thm:bound}
Let $m \geq 2$ be an integer. Any split routing solution to the \emph{Ring Loading Problem} with $m$ split demands can be turned into an unsplittable solution without increasing the load on any edge by more than
\begin{itemize}
\item $\left(1 + \varepsilon\right)D$, if $m \leq 6$ and
\item $\left(\frac{19}{18} + \varepsilon\right)D$, if $m = 7$,
\end{itemize}
for $\varepsilon \leq 5\times 10^{-6}$. Furthermore, there are instances of the \emph{Ring Loading Problem} with $m$ pairwise crossing demands with $L = L^* + D$, for $m \leq 6$, and $L = L^* + \frac{19}{18}D$, for $m = 7$.
\end{theorem}

Note that the dependency on $\varepsilon$ is unavoidable, as our technique for the upper bounds depends on solutions to large mixed integer linear programs that rely on floating point arithmetic. However, the theorem implies that $L \leq L^* + \alpha D$ for instances of the \emph{Ring Loading Problem} where an optimal split routing solution exists such that at most $m$ demands are routed splittably, with $\alpha = 1 + \varepsilon$ if $m \leq 6$ and $\alpha = \frac{19}{18} + \varepsilon$ if $m = 7$.

The remainder of this section is structured as follows. In~\cref{sec:milp_lb}, we introduce the mixed integer linear program that provides us with the upper bounds for $m \leq 7$. In~\cref{sec:it_lb}, we provide a more detailed view on lower bounds, some of which provide the matching lower bounds in~\cref{thm:bound}. On the way, we show a lemma that turns any split routing solution into an instance of the \emph{Ring Loading Problem} while retaining the load increase. We conclude the section by proving~\cref{thm:bound} in~\cref{sec:proof_bound}.

\subsection{Mixed Integer Linear Program}
\label{sec:milp_lb}

We introduce a mixed integer linear program (MILP) that outputs for a given integer $m$ a split routing solution with $m$ demands that cannot be turned into an unsplittable routing without increasing the load on some edge by at least $\alpha D$, for $\alpha \geq 0$ as large as possible. This gives us the claimed upper bounds in~\cref{thm:bound}.

We first introduce our notation. Let $\mathcal{P}$ be the set of all patterns, i.e.\ $z \in \mathcal{P}$ with $z = \left(z_1,\ldots,z_m\right)$ corresponds to a particular choice of values $z_i \in \left\{ v_i, -u_i\right\}$ for all $i \in \left[m\right]$. Note that $\left|\mathcal{P}\right| = 2^m$. Without loss of generality we assume that $D = 1$, as we can otherwise scale the instance. We furthermore assume that every pattern $z \in \mathcal{P}$ starts at $x_p = 0$. We use the following variables: For all $i \in \left[m\right]$, we denote with $u_i$ and $v_i$ the continuous variables that correspond to a particular split routing. For all patterns $z \in \mathcal{P}$ we refer to the maximum and minimum values obtained by the pattern with $a_z$ and $b_z$, respectively. For this purpose we furthermore have binary variables $w_{z,i}^{\min}$ and $w_{z,i}^{\max}$ that indicate at which position $i \in \left\{0,1,\ldots,m\right\}$ of pattern $z \in \mathcal{P}$ the minimal and maximal values are. In fact, $w_{z,i}^{\min} = 1$ if $i \in \argmin_{j \in \left\{0,\ldots,m\right\}} \sum_{k = 1}^j z_{j}$ (and equivalently for $w_{z,i}^{\max}$). The variables $y_z$ correspond to the end point and $c_z$ to the additive performance of pattern $z \in \mathcal{P}$. In order to decide where the maximum value of the additive performance is obtained (see~\cref{obs:obj}) we use a binary variable $w_z$. 

Our MILP formulation is as follows:
\begin{subequations}
\allowdisplaybreaks
\label{eq:milp}
\begin{align}
	 &  & \max\       &  & E                               &                                                               &  &                                                         &  & \notag                \\
	 &  & \text{s.t.} &  & E                               & \leq c_z,                                                     &  & \forall z \in \mathcal{P}                               &  & \label{eq:milp_aux}   \\
	 &  &             &  & u_i + v_i                       & \leq 1,                                                       &  & \forall i \in \left[m\right]                                       &  & \label{eq:milp_feas1} \\
	 &  &             &  & \textstyle\sum_{i = 0}^{m} w_{z,i}^{\min} & \geq 1,                                                       &  & \forall z \in \mathcal{P}                               &  & \label{eq:milp_min}   \\
	 &  &             &  & \textstyle\sum_{i = 0}^{m} w_{z,i}^{\max} & \geq 1,                                                       &  & \forall z \in \mathcal{P}                               &  & \label{eq:milp_max}   \\
	 &  &             &  & a_z                             & \leq \textstyle\sum_{j = 1}^{i} z_{j},                                &  & \forall z \in \mathcal{P}, \forall i \in \left\{0,\ldots,m\right\} &  & \label{eq:milp_aub}   \\
	 &  &             &  & a_z + W_1\cdot\left(1 - w_{z,i}^{\min}\right)                            & \geq \textstyle\sum_{j = 1}^{i} z_{j}, &  & \forall z \in \mathcal{P}, \forall i \in \left\{0,\ldots,m\right\} &  & \label{eq:milp_alb}   \\
	 &  &             &  & b_z                             & \geq \textstyle\sum_{j = 1}^{i} z_{j},                                &  & \forall z \in \mathcal{P}, \forall i \in \left\{0,\ldots,m\right\} &  & \label{eq:milp_blb}   \\
	 &  &             &  & b_z - W_1\cdot\left(1 - w_{z,i}^{\max}\right)                    & \leq \textstyle\sum_{j = 1}^{i} z_{j}, &  & \forall z \in \mathcal{P}, \forall i \in \left\{0,\ldots,m\right\} &  & \label{eq:milp_bub}   \\
	 &  &             &  & y_z                             & = \textstyle\sum_{j = 1}^{m} z_{j},                                    &  & \forall z \in \mathcal{P}                               &  & \label{eq:milp_y}     \\
	 &  &             &  & c_z                             & \geq 2 b_z - y_z,                                              &  & \forall z \in \mathcal{P}                               &  & \label{eq:milp_aplb1} \\
	 &  &             &  & c_z                             & \geq y_z -2 a_z,                                               &  & \forall z \in \mathcal{P}                               &  & \label{eq:milp_aplb2} \\
	 &  &             &  & c_z - W_2 w_z                             & \leq 2 b_z - y_z,                                    &  & \forall z \in \mathcal{P}                               &  & \label{eq:milp_apub1} \\
	 &  &             &  & c_z - W_2 \left(1 - w_z\right)                             & \leq y_z -2 a_z,                               &  & \forall z \in \mathcal{P}                               &  & \label{eq:milp_apub2} \\
	 &  &             &  & u_i, v_i                        & \geq 0,                                                        &  & \forall i \in \left\{0,\ldots,m\right\}                            &  & \label{eq:milp_posd} \\
	 &  &             &  & w_z, w_{z,i}^{\min}, w_{z,i}^{\max}                        & \in \left\{0,1\right\},                                                         &  & \forall z \in \mathcal{P}, \forall i \in \left\{0,\ldots,m\right\}                            &  & \label{eq:milp_bin}
\end{align}
\end{subequations}
Note that the $z_{i}$ values are no variables but place holders for either $-u_i$ or $v_i$ variables (depending on the particular pattern $z \in \mathcal{P}$). The constants $W_1$ and $W_2$ should be large enough such that \cref{eq:milp_alb,eq:milp_bub,eq:milp_apub1,eq:milp_apub2} are trivially satisfied if the respective binary term doesn't vanish. Because we assumed that $D = 1$, we can fix $W_1 = W_2 = m$. The basic idea is that we compute for each pattern $z \in \mathcal{P}$ the additive performance $c_z$ (\crefrange{eq:milp_min}{eq:milp_apub2}) and then ensure that the auxiliary variable $E$ is upper bounded by all of these values (\cref{eq:milp_aux}). The objective then maximizes $E$ in order to find feasible $u$ and $v$ values that maximize the minimum additive performance. More specifically, we first ensure that the split routing is feasible (\cref{eq:milp_feas1,eq:milp_posd}); secondly, we ensure that for each pattern $z \in \mathcal{P}$ the variables $a_z, b_z, y_z$ and $c_z$ are set to the correct values. To see this we will focus on the $a_z$ values, the other variables follow analogously. Let $z \in \mathcal{P}$ be an arbitrary but fixed pattern. We have to show that $a_z = \min_{i \in \left\{0,\ldots,m\right\}} \sum_{j = 1}^{i} z_{j}$. By \cref{eq:milp_aub} we know that $a_z \leq \min_{i \in \left\{0,\ldots,m\right\}} \sum_{j = 1}^{i} z_{j}$. If we now know that there exists an index $i \in \left\{0,\ldots,m\right\}$ such that $a_z \geq \sum_{j = 1}^{i} z_{j}$, the claim follows. This however is ensured by \cref{eq:milp_aub} together with \cref{eq:milp_min}, as there exists at least one index $i$ such that $w_{z,i}^{\min} = 1$ and thus $a_z \geq \sum_{j = 1}^i z_{j} - W_1\left(1 - w_{z,i}^{\min}\right) = \sum_{j = 1}^i z_{j}$.

This very na\"ive approach has many drawbacks, most prominently the massive amount of used binary variables (overall $\left(2m + 1\right)2^m$), the large number of constraints, many symmetries with respect to the $u_i$ and $v_i$ variables and the use of big-M constraints (\cref{eq:milp_alb,eq:milp_bub,eq:milp_apub1,eq:milp_apub2}). However, we only want to solve the MILP for small values of $m$.

The first issue can be improved by the following observation: Assume we are given a fixed pattern $z \in \mathcal{P}$ with the property that there exists an index $0 < i < m$ such that $z_{i} = v_{i}$ and $z_{i +1} = v_{i +1}$. Then the index $i$ will never contribute to a maximum or a minimum of the $a_z$ and $b_z$ variables, as $p_{z}(i -1) \leq p_{z}(i) \leq p_{z}(i+1)$. It is therefore not necessary to maintain the variables $w_{z,i}^{\min}$ and $w_{z,i}^{\max}$. In fact, the binary variable $w_{z,i}^{\min}$ has only to be maintained if $z_{i} = -u_i$ and $z_{i+1} = v_{i +1}$, and likewise $w_{z,i}^{\max}$ has only to be maintained if $z_{i} = v_i$ and $z_{i+1} = -u_{i +1}$, for $0 < i < m$. Similar arguments hold for the border cases $i \in \left\{0,m\right\}$. One can think of these restrictions as local minima and maxima for any fixed pattern. A local minima is $-u_i$ followed by $v_{i +1}$ (excluding the extreme cases), and equivalently a maxima is $v_i$ followed by $-u_{i+1}$. As all local minima and maxima are the same, independent of the specific $u$ and $v$ values, the global minima and maxima (depending on the concrete realization of $u$ and $v$) are obtained at one of those spots. Consequently every pattern has instead of $2m +1$ binary variables at most $m +1$, as each index is associated with at most one of the variables $w_{z,i}^{\min}$ and $w_{z,i}^{\max}$.

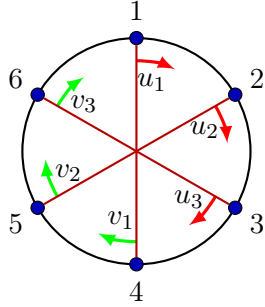
\begin{figure}
\centering
\def\radius{1.5cm}
\def\delangle{25}
\begin{tikzpicture}[label_style/.style={circle, inner sep = 0pt, label distance = 3pt}]


\draw[thick] (0,0) circle (\radius);

\foreach \i in {1,...,6} {
	\node[mynode, label={[label_style]{90 + (\i-1) * -360 / 6}:\i}] (\i) at ({90 - 360 / 6 * (\i-1)}:\radius) {};
}

\foreach \i/\u/\v in {1/u_1/,2/u_2/1,3/u_3/} {
	\draw[split_arc, split_col_u] ({90 - 360 / 6 * (\i -1)}:\radius - 1/5*\radius) arc[radius = \radius - 1/5*\radius, start angle = {90 - 360 / 6 * (\i -1)}, delta angle = -\delangle];

	\node[] at ({90 - 360 / 6 * (\i -1) - \delangle /2}:{\radius - 9/24*\radius}) {$\u$};
}

\foreach \i/\u/\v in {1//v_1,2//v_2,3//v_3} {
	\draw[split_arc, split_col_v] ({270 - 360 / 6 * (\i -1)}:\radius - 1/5*\radius) arc[radius = \radius - 1/5*\radius, start angle = {270 - 360 / 6 * (\i -1)}, delta angle = -\delangle];

	\node[] at ({270 - 360 / 6 * (\i -1) - \delangle /2}:{\radius - 9/24*\radius}) {$\v$};
}

\foreach \i/\j/\d in {1/4/4,2/5/4,3/6/4} {
	\path[mydemand] (\i) to[] (\j);
}
\end{tikzpicture}
\caption{A different view on split routing solutions to visualize necklace symmetries. By relabelling the nodes such that an arbitrary node has label $1$, while remaining nodes are labelled consecutive in clockwise direction, we obtain the same instance with $6$ different labels.}
\label{fig:symm_breaking}
\end{figure}

To break (at least some) symmetries, we added the constraints $u_1 \leq u_i$ and $u_1 \leq v_i$ for all $i \in \left[m\right]$. To see that these inequalities are valid, consider~\cref{fig:symm_breaking}. By relabelling the ring such that an arbitrary but fixed node has label $1$, while the remaining nodes are labelled consecutive in clockwise direction, we can assume that the split demand $u_1$ is the smallest among all other split demands. Note that the number of constraints increased only slightly, while the search space decreased significantly.


\paragraph{Implementation Details and Experiments}

We implemented the MILP with the reduced number of binary variables and the symmetry breaking constraints in C++ and solved it using \texttt{gurobi}~\cite{gurobi}. The computations were carried out on a computer cluster with two Xeon E5-2630 v4 CPU's and $512$GB RAM. Table~\ref{tab:milp} shows results of the experiments together with run times, memory consumption and the number of binary variables after gurobis preprocessing. We can see that the approach produced (almost) optimal solutions for all instances with $m \leq 7$. In particular we see that all split routing solutions of instances with $m \leq 6$ can be turned into unsplittable solutions while increasing the load on any edge by no more than $\left(1 + \varepsilon\right)D$. We furthermore see that all split routing solutions with $m =7$ can be turned into unsplittable solutions while increasing the load on any edge by at most $\left(\frac{19}{18} + \varepsilon\right)D$. As $m$ increases, the amount of used memory and run times are growing rapidly. Our approach is therefore incapable to output a solution for $m \geq 8$.

\begin{table}
\centering
\begin{tabular}{lllll}
\toprule
$m$ & add. perf. & \#bin. var. & time & memory \\
\toprule
$2-6$ & $1 + \varepsilon$ & $<226$ & $<3$ s & $<350$ MB\\
$7$ & $\frac{19}{18} + \varepsilon$ & $528$ & $36$:$25$ h & $105.1$ GB
\end{tabular}
\caption{Additive performance bounds on instances of different sizes together with required computing resources. The numerical error satisfies $\varepsilon \leq 5\times 10^{-6}$.}
\label{tab:milp}
\end{table}

\subsection{Lower Bound Examples}
\label{sec:it_lb}

Skutella disproved in~\cite{MR3463048} with a counterexample Schrijver et al.'s conjecture~\cite{MR1612841} that $L \leq L^* + D$. Particularly, Skutella found an instance of the \emph{Ring Loading Problem} on $16$ nodes and $18$ demands with $D = 10$ and $\delta = \frac{2}{5}$, where the optimum split routing has load $L^* = 39$ while an optimum unsplittable routing has load $L = L^* + D + 1 = 50$, overall providing an additive gap of $\frac{11}{10}D$ (see~\cref{fig:ce_8_10}). In this section we present further counterexamples that will extend the view on known results. We like to highlight that despite our best efforts, we were unable to find counterexamples that improve upon the additive gap of $\frac{11}{10}D$ given by Skutella~\cite{MR3463048}.

Note that in general, a split routing solution where the best additive performance is high with respect to $D$ does not yield a counterexample to Schrijver et al.'s conjecture, as they are not optimum. In fact an optimum split routing for an instance with $m$ pairwise crossing demands splits every demand evenly. Furthermore, the increase of load on some edge by $\alpha D$, for some $\alpha \geq 0$, does not imply that any unsplittable routing increases the load on the edge with the maximum load by $\alpha D$.

However, the following lemma justifies our restriction to split routing solutions in search of instances of the \emph{Ring Loading Problem} where the difference $L - L^*$ is large with respect to $D$. Note that this can be used in conjunction with the split routing solution of Schijver et al.~\cite{MR1612841} to provide a counterexample to their own conjecture, namely an instance of the \emph{Ring Loading Problem} with $L = L^* + \frac{101}{100}D$.

\begin{lemma}
\label{lem:boost}
Let $\alpha \geq 0$. Any split routing solution that cannot be turned into an unsplittable routing without increasing the load on some edge by at least $\alpha D$ can be turned into an instance of the \emph{Ring Loading Problem} with $L - L^* \geq \alpha D$.
\end{lemma}

\def\rectsize{.75cm}
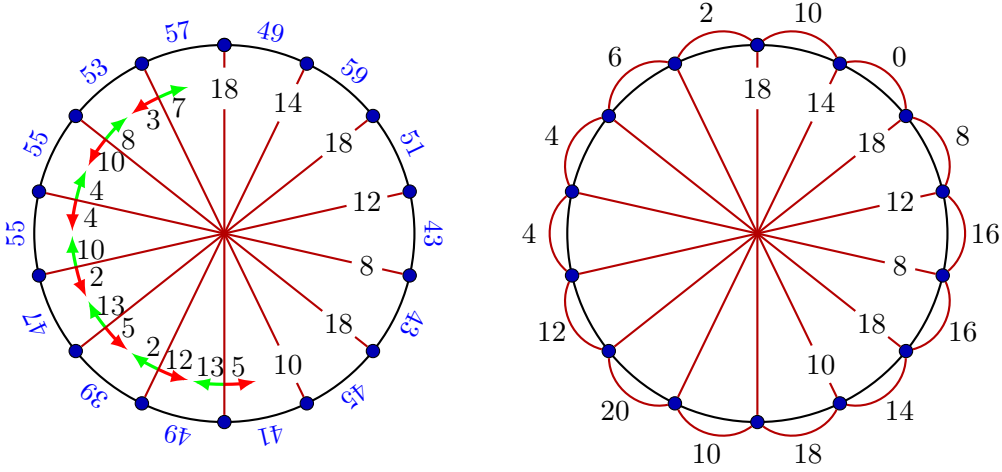
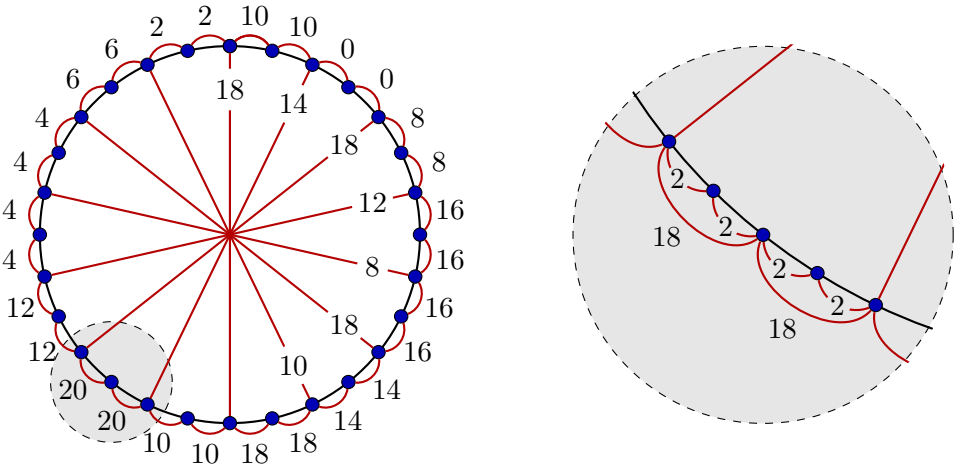
\begin{figure}[htbp]%
\centering
\subfloat[A split routing solution together with incurred edge loads.]{\label{fig:lb_ex_a}\begin{tikzpicture}[
]

\path[] (-\radius - \rectsize, -\radius -\rectsize) rectangle (\radius + \rectsize, \radius + \rectsize);

\draw[thick] (0,0) circle (\radius);

\foreach \i in {0,1,...,14} {
	\node[mynode] (\i) at (90 - 360 / 14 * \i:\radius) {};
}

\foreach \i/\l in {0/49,1/59,2/51,3/43,4/43,5/45,6/41,7/49,8/39,9/47,10/55,11/55,12/53,13/57} {
	\node[text = load_col, rotate = -360 / 28 - 360 / 14 * \i] at (90 - 360 / 28 - 360 / 14 * \i:11/10 * \radius) {\small$\l$};
}

\foreach \i/\u/\v in {0/5/13,1/12/2,2/5/13,3/2/10,4/4/4,5/10/8,6/3/7} {
	\draw[split_arc, split_col_u] (270 - 360 / 14 * \i:\radius - 1/5*\radius) arc[radius = \radius - 1/5*\radius, start angle = 270 - 360 / 14 * \i, delta angle = \delangle];
	\draw[split_arc, split_col_v] (270 - 360 / 14 * \i:\radius - 1/5*\radius) arc[radius = \radius - 1/5*\radius, start angle = 270 - 360 / 14 * \i, delta angle = -\delangle];

	\node[] at ({270 - 360 / 14 * \i + \delangle /2}:{\radius - 7/24*\radius}) {\u};
	\node[] at ({270 - 360 / 14 * \i - \delangle /2}:{\radius - 7/24*\radius}) {\v};
}

\foreach \i/\j/\d in {0/7/18,1/8/14,2/9/18,3/10/12,4/11/8,5/12/18,6/13/10} {
	\path[mydemand] (\i) to[] node[fill = white, circle, inner sep = 1pt, pos=.1] {$\d$} (\j);
}

%
\end{tikzpicture}}%
\quad
\subfloat[Introduction of short demands to equalize the loads to $59$.]{\label{fig:lb_ex_b}\begin{tikzpicture}[
]

\path[] (-\radius - \rectsize, -\radius -\rectsize) rectangle (\radius + \rectsize, \radius + \rectsize);

\draw[thick] (0,0) circle (\radius);

\foreach \i in {0,1,...,14} {
	\node[mynode] (\i) at (90 - 360 / 14 * \i:\radius) {};
}

%
%

\foreach \i/\j/\d in {0/7/18,1/8/14,2/9/18,3/10/12,4/11/8,5/12/18,6/13/10} {
	\path[mydemand] (\i) to[] node[fill = white, circle, inner sep = 1pt, pos=.1] {$\d$} (\j);
}

\foreach \i/\j/\d in {0/1/10,1/2/0,2/3/8,3/4/16,4/5/16,5/6/14,6/7/18,7/8/10,8/9/20,9/10/12,10/11/4,11/12/4,12/13/6,13/0/2} {
	\path[mydemand] (\i) to[bend left = 50] (\j);
}

\foreach \i/\j/\d in {0/1/10,1/2/0,2/3/8,3/4/16,4/5/16,5/6/14,6/7/18,7/8/10,8/9/20,9/10/12,10/11/4,11/12/4,12/13/6,13/14/2} {
	\node[] at ({90 - 360 / 14 * (\j + \i)/2}:{6/5 * \radius}) {$\d$};
}
\end{tikzpicture}}\\[-1ex]
\subfloat[Subdividing edges such that short demands are routed along their edge. The highlighted area contains demands of value $20 > D$.]{\label{fig:lb_ex_c}\begin{tikzpicture}[
]

\path[] (-\radius - \rectsize, -\radius -\rectsize) rectangle (\radius + \rectsize, \radius + \rectsize);

\draw[dashed, fill=white!90!black] (90 - 360 / 28 * 17:\radius) circle (.8cm);

\draw[thick] (0,0) circle (\radius);

\foreach \i in {0,1,...,28} {
	\node[mynode] (\i) at (90 - 360 / 28 * \i:\radius) {};
}

%
%

\foreach \i/\j/\d in {0/14/18,2/16/14,4/18/18,6/20/12,8/22/8,10/24/18,12/26/10} {
	\path[mydemand] (\i) to[] node[fill = white, circle, inner sep = 1pt, pos=.1] {$\d$} (\j);
}

\foreach \i [evaluate=\i as \j using {int{Mod(\i +1,28)}}] in {0,1,...,28} {
  \path[mydemand] (\i) to[bend left = 50] (\j);
}

\foreach \i/\j/\d in {0/1/10,1/2/0,2/3/8,3/4/16,4/5/16,5/6/14,6/7/18,7/8/10,8/9/20,9/10/12,10/11/4,11/12/4,12/13/6,13/14/2} {
\node[] at ({90 - 360 / 14 * (\j + \i)/2 - 360 / 56}:{7/6 * \radius}) {$\d$};
	\node[] at ({90 - 360 / 14 * (\j + \i)/2 + 360 / 56}:{7/6 * \radius}) {$\d$};
}

\end{tikzpicture}}%
\quad
\subfloat[Enlarged section of~\cref{fig:lb_ex_c} showing the subdivision of edges to reduce the demand value of short demands.]{\label{fig:lb_ex_d}\begin{tikzpicture}[scale=3.125
]

\begin{scope}[shift=({90 - 360 / 28 * 17:\radius}),scale=1/3.125]
\path[] (-\radius - \rectsize, -\radius -\rectsize) rectangle (\radius + \rectsize, \radius + \rectsize);
\end{scope}

\draw[dashed, fill=white!90!black] (90 - 360 / 28 * 17:\radius) circle (.8cm);

\clip (90 - 360 / 28 * 17:\radius) circle (.8cm + .5pt);

\draw[thick] (0,0) circle (\radius);

\foreach \i in {0,1,...,28} {
	\node[mynode] (\i) at (90 - 360 / 28 * \i:\radius) {};
}

\node[mynode] (mid1) at (90 - 360 / 28 * 16.5:\radius) {};
\node[mynode] (mid2) at (90 - 360 / 28 * 17.5:\radius) {};

%
%

\foreach \i/\j/\d in {0/14/18,2/16/14,4/18/18,6/20/12,8/22/8,10/24/18,12/26/10} {
	\path[mydemand] (\i) to[] node[fill = white, circle, inner sep = 1pt, pos=.1] {$\d$} (\j);
}

\foreach \i [evaluate=\i as \j using {int{Mod(\i +1,28)}}] in {0,1,...,28} {
  \path[mydemand] (\i) to[bend left = 80] (\j);
}

\path[mydemand] (16) to[bend left = 50] node[pos=.5,fill=white!90!black, circle, inner sep = .25pt] {$2$} (mid1);
\path[mydemand] (mid1) to[bend left = 50] node[pos=.5,fill=white!90!black, circle, inner sep = .25pt] {$2$} (17);

\path[mydemand] (17) to[bend left = 50] node[pos=.5,fill=white!90!black, circle, inner sep = .25pt] {$2$} (mid2);
\path[mydemand] (mid2) to[bend left = 50] node[pos=.5,fill=white!90!black, circle, inner sep = .25pt] {$2$} (18);

   
\foreach \i/\j/\d in {0/1/10,1/2/0,2/3/8,3/4/16,4/5/16,5/6/14,6/7/18,7/8/10,8/9/18,9/10/12,10/11/4,11/12/4,12/13/6,13/14/2} {
  \node[] at ({90 - 360 / 14 * (\j + \i)/2 - 360 / 56}:{10/9 * \radius}) {$\d$};
	\node[] at ({90 - 360 / 14 * (\j + \i)/2 + 360 / 56}:{10/9 * \radius}) {$\d$};
}
\end{tikzpicture}}%
\caption{An example how to create an instance of the \emph{Ring Loading Problem} from a given split routing, where the worst-case load increase is preserved. The split routing has $m = 7$ pairwise crossing demands with $D = 18$ and $\delta = \frac{8}{18}$ such that any unsplittable routing increases the load on some edge by at least $D +1 = 19$.}%
\label{fig:lb_ex}%
\end{figure}

\begin{proof}
An example of the following procedure can be found in~\cref{fig:lb_ex} with a split routing solution on $14$ nodes and $7$ non-zero demands such that any unsplittable routing increases the load on some edge by at least $\frac{19}{18}D$ for $D = 18$. 

Let $l:E \rightarrow \NN$ be the function that maps every edge to its load value with respect to the given split routing solution. We define $l_{\max} := \max_{e \in E} l(e)$ to be the maximum edge load. For every edge $\left\{i,i+1\right\}$ of the ring, we introduce a new demand $d_{i,i+1}$ of value $l_{\max} - l(\left\{ i,i+1\right\})$. We call the edge $\left\{i,i+1\right\}$ \emph{the edge} of demand $d_{i,i+1}$. If all new demands at distance one are routed unsplittably along their edge, we equalize all edge loads, as $d_{i,i+1} + l(\left\{i,i+1\right\}) = l_{\max}$. Because the load is the same on every edge, this configuration is also an optimum split routing for the newly created instance. \Cref{fig:lb_ex_b} shows the made changes to the instance.

There are two issues that can occur. First that not every unsplittable routing to this enhanced instance increases the edge load of an optimum split routing by the same amount as an unsplittable solution for the initial split routing did. This cannot happen, if the newly introduced short demands are routed along their short edge in an optimum unsplittable routing. And second that the introduced demands have a value larger than $D$, which would consequently reduce the maximum increase of edge load relative to $D$. 

The first issue is fixed by the following technique: We subdivide every edge, adding between every neighbouring pair of nodes $i$ and $i+1$ an additional node $i'$. All pairwise crossing demands remain in their current state, whereas the short demands $d_{i,i+1}$ are divided together with the edge. This means we delete the demand $d_{i,i+1}$ and introduce two new demands $d_{i,i'}$ and $d_{i',i+1}$ of the same value (see~\cref{fig:lb_ex_c}). Note that an optimum split routing has the same structure after the subdivision process. We now argue that any unsplittable routing can be turned into a different unsplittable routing where all short demands are routed along their respective edge, without increasing the load on any edge. As the node $i'$ is adjacent to the demands $d_{i,i'}$ and $d_{i',i+1}$ only, ignoring these demands ensures that the load on $\left\{i,i'\right\}$ and $\left\{i',i+1\right\}$ is the same. We consider two cases, first that both $d_{i,i'}$ and $d_{i',i+1}$ are routed the long way, and second that only one of them, say $d_{i,i'}$, is routed the long way while the other demand $d_{i',i+1}$ is routed on its edge. Assume therefore that both $d_{i,i'}$ and $d_{i',i+1}$ are routed the long way. Rerouting both of them to use their edges decreases the load on every edge different form $\left\{i,i'\right\}$ and $\left\{i',i+1\right\}$ by $2d_{i,i+1}$. The load on $\left\{i,i'\right\}$ and $\left\{i',i+1\right\}$ remains the same, implying the claim. If now only $d_{i,i'}$ is routed the long way, rerouting the demand decreases the load on every edge different from $\left\{i,i'\right\}$ by $d_{i,i+1}$, while the load on $\left\{i,i'\right\}$ is now increased to match the new load of $\left\{i',i+1\right\}$. Thus showing that the maximum load did not increase. We finally relabel all nodes from $1$ to $2n$ clockwise along the ring.

We now fix the second issue (see~\cref{fig:lb_ex_d}): We assume in the following, that in an optimum unsplittable routing the short demands are routed along the short edge (see the previous step). Note that only newly introduced demands may have a demand value larger than $D$. Let $d_{j,j+1} > D$ be such a demand. We again subdivide the edge, introducing a node $j'$ between $j$ and $j +1$ on the ring. For the pairwise crossing demands nothing changes. We introduce two new demands $d_{j,j'} = d_{j',j+1} = d_{j,j+1} -D$ and update the original demand value to $d_{j,j+1} = D$. Note that all demand values strictly decreased with respect to the original demand value $d_{j,j+1}$. It furthermore holds that the load on the edges $\left\{j,j'\right\}$ and $\left\{j',j+1\right\}$ is unchanged, if all new demands are routed along the short paths. By the assumption, we know that this is fulfilled for the \enquote{outer} demand $d_{j,j+1}$. For the two created small demands at distance one this follows from the same argumentation as described in the previous modification step. As the invariant is preserved, we can repeat this process until no more demands of value greater than $D$ remain.

Overall we obtain an instance of the \emph{Ring Loading Problem} with $L - L^* \geq \alpha D$, as the load increase of unsplittable routings on some edge by $\alpha D$ guarantees the increase of load by $\alpha D$ on an edge of maximum load.
\end{proof}

Skutella~\cite{MR3463048} found the split routing solution that led to the counterexample (see~\cref{fig:ce_8_10}) by brute force enumeration over instances with specific structural properties on $8$ pairwise crossing integer demands of value at most $D = 10$. We show in the remainder of this section how to modify this split routing solution in order to obtain further counterexamples that are $\delta$-instances for $0 < \delta \leq \frac{1}{2}$ (see also~\cref{fig:bounds}).

\begin{table}[htbp]
\centering
\renewcommand{\arraystretch}{1.1}
\begin{tabular}{l|ccccc}
\cline{2-6}
\rule{0pt}{1.1em}
& cond. &$D$ & $\delta$ & add. perf. & Ref. \\
\toprule
$\begin{aligned}u = (4,4,6,2,7,1,7,2)\ +\\(\varepsilon,\varepsilon,\varepsilon,0,\varepsilon,\varepsilon,\varepsilon,0)\end{aligned}$ & \multirow{2.5}{*}{$\varepsilon \geq 0$} & \multirow{2.5}{*}{$10 + 2\varepsilon$} & \multirow{2.5}{*}{$\dfrac{4}{D}$} & \multirow{2.5}{*}{$D+1 = \left(1 + \dfrac{\delta}{4}\right)D$} & \multirow{2.5}{*}{\cref{fig:ce_8_10}}\\
$\begin{aligned}v = (6,4,4,2,3,7,3,2)\ +\\(\varepsilon,\varepsilon,\varepsilon,0,\varepsilon,\varepsilon,\varepsilon,0)\end{aligned}$ & &  &  &  & \\
\midrule
$\begin{aligned}u = (4,4,6,2,7,1,7,2)\ +\\(\varepsilon,\varepsilon,\varepsilon,\varepsilon,\varepsilon,\varepsilon,\varepsilon,\varepsilon)\end{aligned}$ & \multirow{2.5}{*}{$\varepsilon \in [0,1]$} & \multirow{2.5}{*}{$10 + 2\varepsilon$} & \multirow{2.5}{*}{$\dfrac{4 + 2\varepsilon}{D}$} & \multirow{2.5}{*}{$D +1 = \left(\dfrac{7}{6} - \dfrac{\delta}{6}\right)D$} & \multirow{2.5}{*}{\cref{fig:ce_8_10}}\\
$\begin{aligned}v = (6,4,4,2,3,7,3,2)\ +\\(\varepsilon,\varepsilon,\varepsilon,\varepsilon,\varepsilon,\varepsilon,\varepsilon,\varepsilon)\end{aligned}$ &  & &  &  & \\
\midrule
$u = (7,11,6,10,6,8,5)$ & \multirow{2}{*}{--} & \multirow{2}{*}{$18$} & \multirow{2}{*}{$8/18$} & \multirow{2}{*}{$D+1 =19$} & \multirow{2}{*}{\cref{fig:lb_ex}}\\
$v = (11,3,12,2,2,10,5)$ & & & & &
\end{tabular}
\caption{Lower bound examples and their properties; addition of vectors is component-wise.}
\label{tab:it_res}
\end{table}

Consider the split routing solution in the first row of~\cref{tab:it_res} (see also~\cref{fig:ce_8_10} on the left). For $\varepsilon = 0$, the instance corresponds to one of the split routings given by Skutella~\cite{MR3463048}. By adding $\varepsilon \geq 0$ to certain split demands, the instance changes continuously such that $D = 10 + 2\varepsilon$ and $\delta = \frac{4}{D}$ while any unsplittable solution increases the load on some edge by at least $D +1$. To see that any unsplittable solution increase the load on some edge by at least $11 + 2\varepsilon$, one has to consider all of the $2^8$ different patterns and observe that its additive performance is at least $11 + 2\varepsilon$. In order to obtain the additive performance for a pattern, the minimum $a$, the maximum $b$ and the end point $y$ with respect to $\varepsilon$ has to be calculated and the maximum in~\cref{obs:obj} computed. One problem that occurs, is that the minimum and maximum values considered throughout might not be unique, as the dependence on $\varepsilon$ might change the results. The example in~\cref{fig:generic_lb_ex} illustrates this issue; the minimum $a$ for this pattern varies for different values of $\varepsilon$. The minimum is therefore either $-4 -\varepsilon$ or $-5$. In this particular case however, the additive performance of the pattern is either way at least $11 + 2\varepsilon$, as 
\begin{equation*}
\max \left\{ 2b -y, y - 2a\right\} = 
\begin{cases}\max \left\{ 11 + 2\varepsilon , 7 \right\}\text{, if } a = -5 \text{ and}\\
\max \left\{ 11 + 2\varepsilon, 5 + 2 \varepsilon \right\}\text{, if } a = -4 - \varepsilon,
\end{cases}
\end{equation*}
which is exactly $11 + 2\varepsilon$. A thorough analysis of the remaining cases reveals that the additive performance of all patterns is at least $11 + 2\varepsilon$.

For the second split routing solution given in the second row of~\cref{tab:it_res} (see also~\cref{fig:ce_8_10} on the right) a similar approach can be used to show that the additive performance of all patterns is at least $11 + 2\varepsilon$, where we need the fact that $\varepsilon \in [0,1]$.

\begin{figure}[htbp]
\centering
\begin{tikzpicture}[yscale=.2,xscale=1.35,
node style/.style={draw = black, fill = black, circle, inner sep = 0pt, minimum size = 3pt},
edge style/.style={draw},
line/.style={draw = black},
fwd/.style={draw,thick}]

\definecolor{a_col}{RGB}{0,190,0}
\definecolor{b_col}{RGB}{0,0,190}
\definecolor{c_col}{RGB}{190,0,0}

\node[] at (-.75,6.5) {$p_z(k) = $};

\foreach \c/\e[count=\i] in {0/,-4/-\varepsilon,0/,4/+\varepsilon,2/+\varepsilon,-5/,2/+\varepsilon,-5/,-3/} {
  \node[] (l\i) at (\i cm -1 cm,6.5) {$\c \e$};
}

\foreach \i/\text [evaluate=\i as \j using int(\i+1)] in {0/-u_1,1/+v_2,2/+v_3,3/-u_4,4/-u_5,5/+v_6,6/-u_7,7/+v_8} {
  \draw[edge style,-{latex}] (\i cm +2pt,7.5) to[looseness = 4,out=80,in=100] node[above] {$\text$} 
(\j cm -2pt,7.5);
}

\foreach \offset/\eps/\mi/\epf in {0/0/-5/4,14/2/-6/6} {

\begin{scope}[yshift=-\offset cm]

\node[rotate=90] at (-.75,0) {$\varepsilon = \eps$};


\foreach \y in {\mi,...,\epf}
    \draw[help lines] (8,\y cm) -- (0,\y cm);

\foreach \x in {1,...,8}
    \draw (\x,1pt) -- (\x,-3pt) node[anchor=north] {\x};

\path[line, {latex}-] (8.5,0) to[] node[anchor = north west,pos = 0] {$k$} (0,0) node[anchor=east] {$0$};

\path[line, -{latex}] (0,\mi cm -.5 cm) to[] (0,5 cm + \eps cm);

\ifthenelse{\eps = 2}
{
\node[draw=red,inner sep = 0pt, minimum size = 6pt, circle, very thick] at (1,-6) {};
}
{
\node[draw=red,inner sep = 0pt, minimum size = 6pt, circle, very thick] at (5,-5) {};
\node[draw=red,inner sep = 0pt, minimum size = 6pt, circle, very thick] at (7,-5) {};
}

\node[node style] (0) at (0,0) {};
\node[node style] (1) at (1,-4 cm -\eps cm) {};
\node[node style] (2) at (2,0) {};
\node[node style] (3) at (3,4 cm + \eps cm) {};
\node[node style] (4) at (4,2 cm + \eps cm) {};
\node[node style] (5) at (5,-5) {};
\node[node style] (6) at (6,2 cm + \eps cm) {};
\node[node style] (7) at (7,-5) {};
\node[node style] (8) at (8,-3) {};

\foreach \i [evaluate=\i as \j using int(\i+1)] in {0,...,7} {
  \path[edge style, fwd, a_col] (\i) to[] (\j);
}
\end{scope}
}
\end{tikzpicture}
\caption{A pattern for the unsplittable solution $z = (-u_1,v_2,v_3,-u_4,-u_5,v_6,-u_7,v_8)$ of the split routing solution in the first row of~\cref{tab:it_res} with $\varepsilon = 0$ (top) and $\varepsilon = 2$ (bottom). The minima of the two variations (highlighted circles) are different.}
\label{fig:generic_lb_ex}
\end{figure}
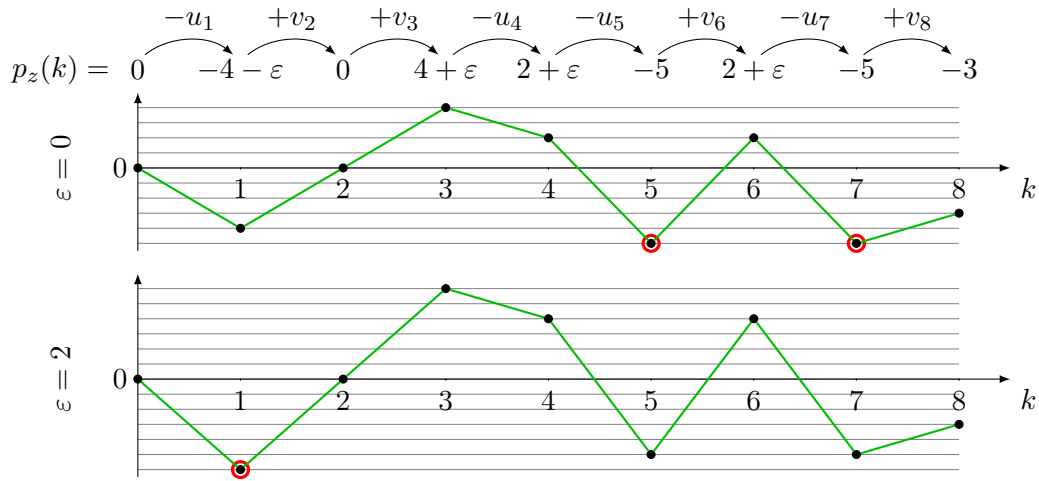


\def\radius{2.85cm}
\def\radiusoffset{5pt}
\def\delangle{11}

\def\radius{3cm}
\def\delangle{11}
\begin{figure}
\centering
\begin{tikzpicture}[
]

\draw[thick] (0,0) circle (\radius);

\foreach \i in {0,1,...,16} {
	\node[mynode] (\i) at (90 - 360 / 16 * \i:\radius) {};
}

\foreach \i/\l in {0/29,1/29,2/31,3/31,4/35,5/29,6/33,7/33,8/35,9/35,10/33,11/33,12/29,13/35,14/31,15/31} {
}

\foreach \i/\u/\v/\eps/\p in {0/4/6/\varepsilon/+,1/4/4/\varepsilon/+,2/6/4/\varepsilon/+,3/2/2//,4/7/3/\varepsilon/+,5/1/7/\varepsilon/+,6/7/3/\varepsilon/+,7/2/2//} {
	\draw[split_arc, split_col_u] (270 - 360 / 16 * \i:\radius - 1/24*\radius) arc[radius = \radius - 1/24*\radius, start angle = 270 - 360 / 16 * \i, delta angle = \delangle];
	\draw[split_arc, split_col_v] (270 - 360 / 16 * \i:\radius - 1/24*\radius) arc[radius = \radius - 1/24*\radius, start angle = 270 - 360 / 16 * \i, delta angle = -\delangle];

	\node[] at ({270 - 360 / 16 * \i + \delangle /2}:{\radius - 3/24*\radius}) {\small$\u$};
	\node[] at ({270 - 360 / 16 * \i - \delangle /2}:{\radius - 3/24*\radius}) {\small$\v$};
  
  \node[] at ({270 - 360 / 16 * \i + \delangle /2}:{\radius - 5.5/24*\radius}) {\small$\p$};
	\node[] at ({270 - 360 / 16 * \i - \delangle /2}:{\radius - 5.5/24*\radius}) {\small$\p$};
  
  \node[] at ({270 - 360 / 16 * \i + \delangle /2}:{\radius - 8/24*\radius}) {\small$\eps$};
 	\node[] at ({270 - 360 / 16 * \i - \delangle /2}:{\radius - 8/24*\radius}) {\small$\eps$};
}

\foreach \i/\j/\d/\eps/\rot in {0/8/10/\!+\!2\varepsilon/180,1/9/8/\!+\!2\varepsilon/0,2/10/10/\!+\!2\varepsilon/0,3/11/4//0,4/12/10/\!+\!2\varepsilon/0,5/13/8/\!+\!2\varepsilon/0,6/14/10/\!+\!2\varepsilon/0,7/15/4//0} {
	\path[mydemand] (\i) to[] node[fill = white, sloped, rotate = \rot, inner sep = 2pt, pos=.15] {$\d\eps$} (\j);
}



\end{tikzpicture}\qquad
\begin{tikzpicture}[
]

\draw[thick] (0,0) circle (\radius);

\foreach \i in {0,1,...,16} {
	\node[mynode] (\i) at (90 - 360 / 16 * \i:\radius) {};
}

\foreach \i/\l in {0/29,1/29,2/31,3/31,4/35,5/29,6/33,7/33,8/35,9/35,10/33,11/33,12/29,13/35,14/31,15/31} {
}

\foreach \i/\u/\v/\eps/\p in {0/4/6/\varepsilon/+,1/4/4/\varepsilon/+,2/6/4/\varepsilon/+,3/2/2/\varepsilon/+,4/7/3/\varepsilon/+,5/1/7/\varepsilon/+,6/7/3/\varepsilon/+,7/2/2/\varepsilon/+} {
	\draw[split_arc, split_col_u] (270 - 360 / 16 * \i:\radius - 1/24*\radius) arc[radius = \radius - 1/24*\radius, start angle = 270 - 360 / 16 * \i, delta angle = \delangle];
	\draw[split_arc, split_col_v] (270 - 360 / 16 * \i:\radius - 1/24*\radius) arc[radius = \radius - 1/24*\radius, start angle = 270 - 360 / 16 * \i, delta angle = -\delangle];

	\node[] at ({270 - 360 / 16 * \i + \delangle /2}:{\radius - 3/24*\radius}) {\small$\u$};
	\node[] at ({270 - 360 / 16 * \i - \delangle /2}:{\radius - 3/24*\radius}) {\small$\v$};
  
  \node[] at ({270 - 360 / 16 * \i + \delangle /2}:{\radius - 5.5/24*\radius}) {\small$\p$};
	\node[] at ({270 - 360 / 16 * \i - \delangle /2}:{\radius - 5.5/24*\radius}) {\small$\p$};
  
  \node[] at ({270 - 360 / 16 * \i + \delangle /2}:{\radius - 8/24*\radius}) {\small$\eps$};
 	\node[] at ({270 - 360 / 16 * \i - \delangle /2}:{\radius - 8/24*\radius}) {\small$\eps$};
}

\foreach \i/\j/\d/\eps/\rot in {0/8/10/\!+\!2\varepsilon/180,1/9/8/\!+\!2\varepsilon/0,2/10/10/\!+\!2\varepsilon/0,3/11/4/\!+\!2\varepsilon/0,4/12/10/\!+\!2\varepsilon/0,5/13/8/\!+\!2\varepsilon/0,6/14/10/\!+\!2\varepsilon/0,7/15/4/\!+\!2\varepsilon/0} {
	\path[mydemand] (\i) to[] node[fill = white, sloped, rotate = \rot, inner sep = 2pt, pos=.15] {$\d\eps$} (\j);
}



\end{tikzpicture}
\caption{Two examples of split routings with $8$ pairwise crossing demands and $D = 10 + 2\varepsilon$ such that any unsplittable routing increases the load on some edge by at least $D + 1 = 11 + 2\varepsilon$, for $\varepsilon \geq 0$ (left) and $\varepsilon \in [0,1]$ (right). The instances coincide with the split routing from Skutella~\cite{MR3463048} for $\varepsilon = 0$.}
\label{fig:ce_8_10}
\end{figure}


\subsection{Proof of Theorem~\ref{thm:bound}}
\label{sec:proof_bound}

With~\cref{sec:milp_lb,sec:it_lb} at hand, we can turn our attention to the proof of~\cref{thm:bound}.

\begin{proof}[Proof of~\cref{thm:bound}]
We see in~\cref{tab:milp} that any split routing solution with $m \leq 6$ can be turned into an usplittable solution while increasing the load on any edge by no more than $\left(1 + \varepsilon\right)D$ (after rescaling by $D$). For $m = 7$ the worst-case increase of load is $\left(\frac{19}{18} + \varepsilon\right)D$. Thus the upper bounds follow. The lower bounds can be produced with~\cref{lem:boost} together with the found split routing solutions given in~\cref{fig:ce_3_4,fig:ce_5_4,fig:ce_6_2,fig:lb_ex}. The remaining even cases follow with the same type of instance as~\cref{fig:ce_6_2}.
\end{proof}

\begin{figure}
\centering
\subfloat[$m = 3$, $D = 4$]{\label{fig:ce_3_4}\def\radius{1.4cm}\def\delangle{25}\begin{tikzpicture}[
]

\path[] (0,-3) -- (0,3);

\draw[thick] (0,0) circle (\radius);

\foreach \i in {0,1,...,6} {
	\node[mynode] (\i) at (90 - 360 / 6 * \i:\radius) {};
}

\foreach \i/\l in {0/4,1/6,2/8,3/8,4/6,5/4} {
	\node[text = load_col, rotate = -360 / 12 - 360 / 6 * \i] at (90 - 360 / 12 - 360 / 6 * \i:9/8 * \radius) {\small$\l$};
}

\foreach \i/\u/\v in {0/2/2,1/3/1,2/3/1} {
	\draw[split_arc, split_col_u] (270 - 360 / 6 * \i:\radius - 1/5*\radius) arc[radius = \radius - 1/5*\radius, start angle = 270 - 360 / 6 * \i, delta angle = \delangle];
	\draw[split_arc, split_col_v] (270 - 360 / 6 * \i:\radius - 1/5*\radius) arc[radius = \radius - 1/5*\radius, start angle = 270 - 360 / 6 * \i, delta angle = -\delangle];

	\node[] at ({270 - 360 / 6 * \i + \delangle /2}:{\radius - 7/24*\radius}) {\u};
	\node[] at ({270 - 360 / 6 * \i - \delangle /2}:{\radius - 7/24*\radius}) {\v};
}

\foreach \i/\j/\d in {0/3/4,1/4/4,2/5/4} {
	\path[mydemand] (\i) to[] node[fill = white, circle, inner sep = 1pt, pos=.1] {$\d$} (\j);
}



\end{tikzpicture}}%
\hfill
\subfloat[$m = 5$, $D = 4$]{\label{fig:ce_5_4}\def\radius{2.1cm}\def\delangle{17}\begin{tikzpicture}[
]

\path[] (0,-3) -- (0,3);

\draw[thick] (0,0) circle (\radius);

\foreach \i in {0,1,...,10} {
	\node[mynode] (\i) at (90 - 360 / 10 * \i:\radius) {};
}

\foreach \i/\l in {0/8,1/8,2/8,3/10,4/12,5/12,6/12,7/12,8/10,9/8} {
	\node[text = load_col, rotate = -360 / 20 - 360 / 10 * \i] at (90 - 360 / 20 - 360 / 10 * \i:11/10 * \radius) {\small$\l$};
}

\foreach \i/\u/\v in {0/2/2,1/2/2,2/2/2,3/3/1,4/3/1} {
	\draw[split_arc, split_col_u] (270 - 360 / 10 * \i:\radius - 1/5*\radius) arc[radius = \radius - 1/5*\radius, start angle = 270 - 360 / 10 * \i, delta angle = \delangle];
	\draw[split_arc, split_col_v] (270 - 360 / 10 * \i:\radius - 1/5*\radius) arc[radius = \radius - 1/5*\radius, start angle = 270 - 360 / 10 * \i, delta angle = -\delangle];

	\node[] at ({270 - 360 / 10 * \i + \delangle /2}:{\radius - 7/24*\radius}) {\u};
	\node[] at ({270 - 360 / 10 * \i - \delangle /2}:{\radius - 7/24*\radius}) {\v};
}

\foreach \i/\j/\d in {0/5/4,1/6/4,2/7/4,3/8/4,4/9/4} {
	\path[mydemand] (\i) to[] node[fill = white, circle, inner sep = 1pt, pos=.1] {$\d$} (\j);
}



\end{tikzpicture}}%
\hfill
\subfloat[$m = 6$, $D = 2$]{\label{fig:ce_6_2}\def\radius{2.2cm}\def\delangle{14}\begin{tikzpicture}[
]

\path[] (0,-3) -- (0,3);

\draw[thick] (0,0) circle (\radius);

\foreach \i in {0,1,...,12} {
	\node[mynode] (\i) at (90 - 360 / 12 * \i:\radius) {};
}

\foreach \i/\l in {0/6,1/6,2/6,3/6,4/6,5/6,6/6,7/6,8/6,9/6,10/6,11/6} {
	\node[text = load_col, rotate = -360 / 24 - 360 / 12 * \i] at (90 - 360 / 24 - 360 / 12 * \i:11/10 * \radius) {\small$\l$};
}

\foreach \i/\u/\v in {0/1/1,1/1/1,2/1/1,3/1/1,4/1/1,5/1/1} {
	\draw[split_arc, split_col_u] (270 - 360 / 12 * \i:\radius - 1/5*\radius) arc[radius = \radius - 1/5*\radius, start angle = 270 - 360 / 12 * \i, delta angle = \delangle];
	\draw[split_arc, split_col_v] (270 - 360 / 12 * \i:\radius - 1/5*\radius) arc[radius = \radius - 1/5*\radius, start angle = 270 - 360 / 12 * \i, delta angle = -\delangle];

	\node[] at ({270 - 360 / 12 * \i + \delangle /2}:{\radius - 7/24*\radius}) {\u};
	\node[] at ({270 - 360 / 12 * \i - \delangle /2}:{\radius - 7/24*\radius}) {\v};
}

\foreach \i/\j/\d in {0/6/2,1/7/2,2/8/2,3/9/2,4/10/2,5/11/2} {
	\path[mydemand] (\i) to[] node[fill = white, circle, inner sep = 1pt, pos=.1] {$\d$} (\j);
}



\end{tikzpicture}}%
\caption{Three split routing solutions with $m$ pairwise crossing demands of value $D$.}
\end{figure}

\section{Conclusions}
\label{sec:concl}

We showed that any split routing solution to the \emph{Ring Loading Problem} can be turned into an unsplittable solution while increasing the load on any edge by at most $\frac{13}{10}D$. We furthermore showed that split routing solutions with at most $7$ pairwise crossing demands cannot yield lower bounds with additive performance worse than $\left(\frac{19}{18} + \varepsilon\right)D$, for a small $\varepsilon$. On the way, we also proved that any split routing solution with large additive performance can be turned into an instance of the \emph{Ring Loading Problem} while maintaining the load increase. We also gave a broader view on instances where the difference $L - L^*$ is large with respect to $D$.

The obvious open problem is the correct value of additive load increase. Skutella~\cite{MR3463048} conjectured that $L \leq L^* + \frac{11}{10}D$, which is matched by the currently best lower bound instance. After spending numerous hours on finding a stronger lower bound, unfortunately without any success, we are tempted to believe that this conjecture might be true. In any case, we highly doubt that the current best upper bound is the definitive answer.

\section{Acknowledgements}

We thank Martin Skutella for introducing us to the \emph{Ring Loading Problem} and for many fruitful discussions and comments. We also thank Torsten Mütze for reading an early draft of the paper.

\bibliographystyle{alpha}
\bibliography{ref}
\end{document}